\documentclass[11pt]{article}

\usepackage[english]{babel}

\usepackage[letterpaper,top=2cm,bottom=2cm,left=2cm,right=2cm,marginparwidth=1.75cm]{geometry}

\mathchardef\mhyphen="2D 

\newcommand{\dclass}{\mathrm{D}} 
\newcommand{\rclass}{\mathrm{R}} 
\newcommand{\qclass}{\mathrm{Q}} 
\newcommand{\bs}{\mathrm{bs}} 
\newcommand{\fbs}{\mathrm{fbs}} 
\newcommand{\s}{\mathrm{s}} 
\newcommand{\degree}{\mathrm{deg}} 

\newcommand{\cer}{\mathrm{C}} 
\newcommand{\cmin}{\mathrm{C_{min}}} 
\newcommand{\cminstar}{\mathrm{C_{min}^*}} 
\newcommand{\rank}{\mathrm{rank}} 
\newcommand{\alt}{\mathrm{alt}} 
\newcommand{\stripe}{\mathrm{STRIPE}}
\newcommand{\cP}{\mathcal{P}}
\newcommand{\G}{\mathcal{G}}
\newcommand{\codim}{\mathsf{codim}}
\newcommand{\CC}{\mathsf{CC}}

\newcommand{\subcubedt}{\mathrm{D^{sc}}} 
\newcommand{\ind}{\mathsf{IND}}
\newcommand{\pdt}{\mathsf{PDT}}

\usepackage{amssymb}

\usepackage{comment}
\usepackage{amsmath, amsthm}
\usepackage{graphicx}
\usepackage{mathtools}
\usepackage[ruled,vlined]{algorithm2e}
\usepackage[noend]{algpseudocode}
\LinesNumbered
\usepackage[colorlinks=true, allcolors=blue]{hyperref}
\newtheorem{question}{Question}

\usepackage{xcolor}

\newtheorem{theorem}{Theorem}[section]
\newtheorem{definition}[theorem]{Definition}

\newtheorem{lemma}[theorem]{Lemma}

\newtheorem{observation}[theorem]{Observation}
\newtheorem{proposition}[theorem]{Proposition}

\newtheorem*{proof*}{Proof of \hspace{-10pt}}

\newtheorem{open}[theorem]{{\color{red}Question}}
\usepackage{cleveref}
\usepackage{newverbs}
\date{}

\title{Decision Tree Complexity versus Block Sensitivity and Degree}

\author{
    Rahul Chugh\\[1mm]
    \small IIT Kharagpur,\\ 
    \small \texttt{rahulchugh@iitkgp.ac.in}
    \and
    Supartha Podder\footnote{S.P. is supported by US Department of Energy (grant no DE-SC0023179) and partially supported by US National Science Foundation (award no 1954311).}\\[1mm]
     \small Stony Brook University.\\
    \small \texttt{supartha@cs.stonybrook.edu}
    \and
    Swagato Sanyal\footnote{S.S. is supported by an ISIRD grant by Sponsored Research and Industrial Consultancy, IIT Kharagpur.}\\[1mm]
    \small IIT Kharagpur,\\ 
    \small \texttt{swagato@cse.iitkgp.ac.in}
    }

\begin{document}
\maketitle

\begin{abstract}
Relations between the decision tree complexity and various other complexity measures of Boolean functions is a thriving topic of research in computational complexity. While decision tree complexity is long known to be polynomially related with many other measures, the optimal exponents of many of these relations are not known. It is known that decision tree complexity is bounded above by the cube of block sensitivity, and the cube of polynomial degree. However, the widest separation between decision tree complexity and each of block sensitivity and degree that is witnessed by known Boolean functions is quadratic.

Proving quadratic relations between these measures would resolve several open questions in decision tree complexity. For example, we get a tight relation between decision tree complexity and square of randomized decision tree complexity and a tight relation between zero-error randomized decision tree complexity and square of fractional block sensitivity, resolving an open question raised by Aaronson \cite{aaronson2008quantum}. In this work, we investigate the tightness of the existing cubic upper bounds. 

We improve the cubic upper bounds for many interesting classes of Boolean functions. We show that for graph properties and for functions with a constant number of alternations, both of the cubic upper bounds can be improved to quadratic. We define a class of Boolean functions, which we call the zebra functions, that comprises Boolean functions where each monotone path from $0^n$ to $1^n$ has an equal number of alternations. This class contains the symmetric and monotone functions as its subclasses. We show that for any zebra function, decision tree complexity is at most the square of block sensitivity, and certificate complexity is at most the square of degree. 

Finally, we show using a lifting theorem of communication complexity by G{\"{o}}{\"{o}}s, Pitassi and Watson that the task of proving an improved upper bound on the decision tree complexity for all functions is in a sense equivalent to the potentially easier task of proving a similar upper bound on communication complexity for each bi-partition of the input variables, for all functions. In particular, this implies that to bound the decision tree complexity it suffices to bound smaller measures like parity decision tree complexity, subcube decision tree complexity and decision tree rank, that are defined in terms of models that can be efficiently simulated by communication protocols.
\end{abstract}
\section{Introduction}
The objective of computational complexity theory is to determine the amount of resources needed to solve various computational problems. A central goal of computational complexity theory is to prove limitations of general computational models meant to capture arbitrary computations that can be realized physically. Examples of such models are Turing machines and Boolean circuits. However, proving unconditional negative results for such general models seems far beyond the current reach of researchers. Hence, a popular theme of research has been to analyze more restricted models of computation, whose examples include Boolean circuits of bounded depth, communication complexity and decision trees. Such pursuits have borne fruits, and may serve as a stepping stone to build up towards a more complete understanding of general models.

This paper deals with decision tree. Decision tree is amongst the simplest models of computation. Unlike the more general models, it is often within the reach of researchers to determine the complexity of various interesting tasks in this model. A rich body of work has emerged centering the study of decision tree complexity, its variants and their connections with many other measures of complexity. This pursuit has resulted in a large number of measures having been defined and related to decision tree complexity. Interestingly, decision tree complexity and all of these measures have been shown to be polynomially related to one another. We refer the readers to the survey of  Buhrman and de Wolf \cite{BdeW02Survey} in this regard. Attention has focussed on determining the exact polynomial dependence between various measures. Table $1$ in \cite{ABKRT21} records the current status of our knowledge of these relations.

We now informally define the measures relevant to this work. The formal definitions may be found in \Cref{sec:exprelims}. Let $f$ be a Boolean function and $x$ be a generic input string. A decision tree that computes $f$ is an algorithm that queries various input bits of an input string $x$, possibly adaptively, to compute $f(x)$. The complexity of a decision tree $\mathcal{A}$ that computes $f$, is the maximum number of input bits that $\mathcal A$ queries on any $x$ (possibly adaptively) to compute $f(x)$. The \emph{deterministic decision tree complexity} of $f$, denoted by $\dclass(f)$, is the minimum complexity of any decision tree that computes $f$. A subset $S \subseteq [n]$ of indices is called a \emph{sensitive block} for $x$ if flipping all bits in $x$ with indices in $S$ flips the value of $f$. The block sensitivity of $f$ on $x$, denoted by $\bs(f, x)$, is the maximum number of disjoint sensitive blocks in $x$. The block sensitivity $\bs(f)$ of $f$ is $\max_{x} \bs(f, x)$. $\bs(f, x)$ can be formulated as the value of an integer linear program. The value of its fractional relaxation is the fractional block sensitivity $\fbs(f, x)$ on $x$. The fractional block sensitivity $\fbs(x)$ is defined to be $\max_{x} \fbs(f, x)$. The exact degree $\degree(f)$ of $f$ is the degree of the unique multi-linear polynomial $P_f$ that computes $f$. Let $f(x)=b$. Then the $b$-certificate complexity of $f$ on $x$, denoted by $\cer_b(f, x)$, is the minimum co-dimension of a subcube $\mathcal{C}$ that contains $x$, such that $f(y)=b$ for every $y \in \mathcal{C}$. The \emph{certificate complexity} (resp. \emph{minimum certificate complexity}) of $f$, denoted by $\cer(f)$ (resp. $\cmin(f)$), is $\max_x\{C_{f(x)}(f, x)\}$ (resp. $\min_x\{C_{f(x)}(f, x)\}$). We define $\cminstar(f)$ to be $\max_\mathcal{C} \cmin(f\mid_\mathcal{C})$, where $f\mid_\mathcal{C}$ denotes the restriction of $f$ to $\mathcal{C}$, and the maximum is over all subcubes of the domain $\{0,1\}^n$.

It is known (see Table $1$ in \cite{ABKRT21}; also follows from \Cref{prop:basic_facts} parts (2), (3), (4) and (5)) that for all Boolean functions $f$, $\dclass(f)=O(\bs(f)^3)$ and $\dclass(f)=O(\degree(f)^3)$. It is also known that there exist Boolean functions $g, h$ such that $\dclass(g)=\Omega(\bs(g)^2)$ and $\dclass(h)=\Omega(\degree(h)^2)$. In light of the above, the following is a natural question.
\begin{question}
\label[question]{qn:q1}
Are the upper bounds (1) $\dclass(f)=O(\bs(f)^3)$ and (2) $\dclass(f)=O(\degree(f)^3)$ asymptotically tight?
\end{question}
Question~\ref{qn:q1} forms the center of our study, with a special focus on the possibility of a negative answer. If the first bound can be improved to $\dclass(f) = O(\bs(f)^2)$ to match the widest separation known, then it would resolve several open questions in decision tree complexity. Most notably, it would imply the following tight\footnote{Some of the results will be tight up to poly-logarithmic factors} relations: $\dclass(f) = O(\rclass(f)^2)$ and quantum decision tree complexity is upper bounded by the square of fractional block sensitivity i.e., $\qclass(f) = \tilde{O}(fbs(f)^2)$.
It will also imply the tight relation: $\rclass_0(f) = O(\fbs(f)^2)$, resolving an open question by Aaronson 
\cite{aaronson2008quantum}.
Similarly, if $\dclass(f) = O(\degree(f)^2)$ is true, it will also improve several relations among  measures in decision tree complexity e.g., the current best upper bound of $\dclass(f) =O(\s(f)^6)$ would improve to $\dclass(f) =O(\s(f)^4)$, and we will achieve the tight relation $\dclass(f) =O(\qclass_E(f)^2)$, where $\qclass_E$ is the exact quantum decision tree complexity.
\subsection{Our results}
We now give an overview of our results. We divide our results into two parts.
\subsubsection{Improved upper bounds for classes of functions}
Our first category of results studies Question~\ref{qn:q1} for various classes of Boolean functions.
\paragraph{Graph properties}
A graph property is a Boolean function where the input bit string is the adjacency matrix of an undirected graph, which is invariant under permutations of the vertices. The input to a graph property of graphs with $n$ vertices is a string in $\{0,1\}^{n \choose 2}$. Each input bit is indexed by a (unordered) pair of vertices. The value of the bit indicates whether or not there is an edge between those vertices. Decision tree complexity of graph properties has seen a long and rich line of research \cite{R73, K74, RV76, BBL74, SW86, Y87, K88, H91, CK07}.

We provide strong negative answers to Questions~\ref{qn:q1} (1) and~\ref{qn:q1} (2) for graph properties.
\begin{theorem}
\label[theorem]{thm:graph_P}
For any graph property $\cP$ of graphs with $n$ vertices, (1) $\dclass(\cP) = O(\bs(\cP))^2$ and (2) $\dclass(\cP)=O(\degree(\cP))^2$.
\end{theorem}
\Cref{thm:graph_P} follows immediately from the following lemma and the observation that $\dclass(\cP)\leq\binom{n}{2}$.
\begin{lemma}
\label[lemma]{lemma:graph_key}
For any non-trivial graph property $\cP$ on any graph $G=(V,E)$ with $|V|=n$, (1) $\bs(\cP) = \Omega(n)$ and (2) $\degree(\cP) = \Omega(n)$.
\end{lemma}
\Cref{lemma:graph_key} (1) and (2) are both asymptotically tight. Their tightness is witnessed by graph properties \emph{SINK} \cite{R73} and \emph{SCORPION GRAPH} \cite{BBL74}. For both of these graph properties, $\bs(\cdot), \degree(\cdot) \leq\dclass(\cdot)=O(n)$.

Graph property is a subclass of the rich and important class of \emph{transitive Boolean functions}, which is the class of functions that are invariant under the action of some transitive group acting on the variables. Various complexity measures of transitive functions have been analyzed in the literature. Sun \cite{S07} showed that for any non-trivial transitive Boolean function $f$ on $N$ variables, $\bs(f)=\Omega(N^{1/3})$. Amano \cite{A11} constructed a transitive Boolean function $f$ on $N$ bits for which $\bs(f)=O(N^{3/7})$ holds. Thus, the statement of \Cref{lemma:graph_key} (1) does not hold for arbitrary transitive functions, and the underlying graph structure and the invariance under permutations of vertices are crucial for the lemma to hold. Kulkarni and Tal \cite{KT16} showed that fractional block sensitivity of any non-trivial transitive function on $N$ variables is $\Omega(\sqrt{N})$. \Cref{lemma:graph_key} (1) proves the same lower bound on the smaller measure of block sensitivity (\Cref{prop:basic_facts} (1)) for the special case of graph properties. We remark that for the special case of non-trivial monotone graph properties, the bound in \Cref{lemma:graph_key} (2) can be improved to $\Omega(n^2)$ \cite{RV76, DK99}.

\paragraph{Zebra functions} We say that a function is a \emph{zebra function} if in each monotone path in the Boolean hypercube from $0^n$ to $1^n$, the function value changes equal number of times. The class of zebra functions contains symmetric and monotone Boolean functions as two important sub-classes. However, there are interesting zebra functions that are neither monotone nor symmetric. One example is the \emph{Kushilevitz's function} (Footnote $1$ on page $560$ of \cite{NW95}), for which the value of the function is determined by the Hamming weight of the input, unless the Hamming weight is $3$ \cite{W22}. Another set of examples is functions expressible as $g(\lceil \ell(x_1,\ldots,x_n)\rceil)$, where $g:\mathbf{Z}\rightarrow\{0,1\}$ is arbitrary, and $\ell(x_1,\ldots,x_n):=\ell_0+\sum_{i=1}^n \ell_ix_i$ is an affine function where each coefficient $\ell_i$ lies in the interval $(0, 1)$. Visually, a zebra function induces a partition of the Boolean hypercube into monochromatic``stripes" (see \Cref{fig:zebra}). The value of the function is constant in each stripe, and different in adjacent stripes. We are optimistic that zebra functions will turn out to be useful in future research in Boolean function complexity.

\begin{figure}
    \centering
    \includegraphics[width=150pt]{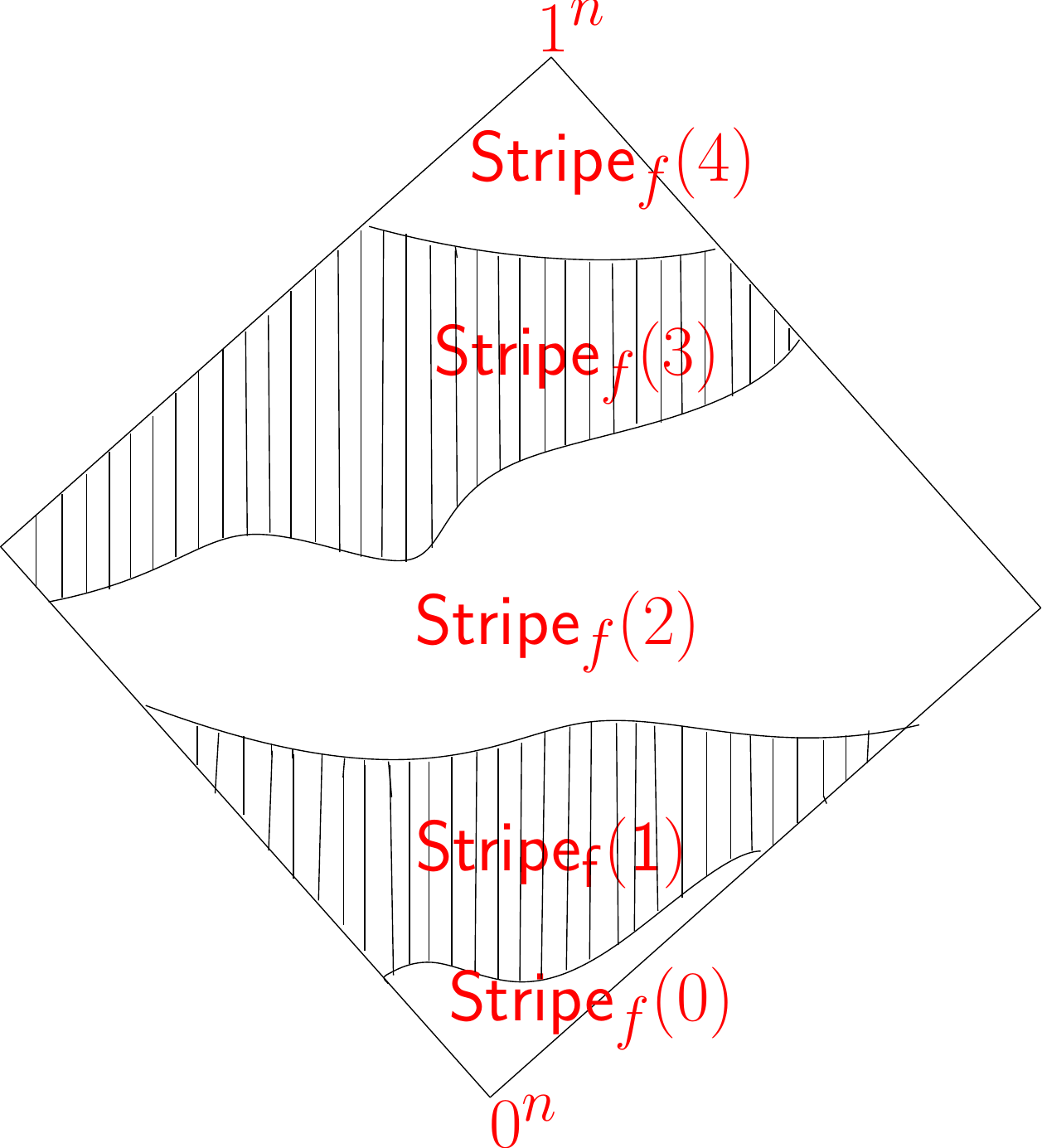}
    \caption{A zebra function with 5 stripes}
    \label[figure]{fig:zebra}
\end{figure}

We prove the following theorem.
\begin{theorem}
\label[theorem]{thm:zebra-d-bs-deg}
For every zebra function $f$, 
\begin{enumerate}
    \item $\dclass(f) = O(\bs(f)^2)$ and 
    \item $\cer(f) =O(\degree(f)^2 )$.
\end{enumerate}
\end{theorem}
\Cref{thm:zebra-d-bs-deg} (1) gives a negative answer to Question~\ref{qn:q1} (1) for zebra functions. \Cref{thm:zebra-d-bs-deg} (1) is tight, as witnessed by the monotone function $\mathsf{TRIBES}_{\sqrt{n} \times \sqrt{n}}$ (see \Cref{def:tribes}). As for \Cref{thm:zebra-d-bs-deg} (2), we do not know if there exists a zebra function $f$ for which $\cer(f)$ is polynomially larger than $\degree(f)$. If $\cer(\cdot)=O(\degree(\cdot))$ holds for all zebra functions, then it follows from \Cref{prop:basic_facts} (5) that $\dclass(\cdot)=O(\degree(\cdot)^2)$ holds for all zebra functions. We leave the question of whether $\dclass(f)=O(\degree(f))^2$ is true for zebra functions $f$ as an open question.
\paragraph{Functions with bounded alternation number}
The alternation number $\alt(f)$ of a Boolean function $f$ is the maximum number of times the function value changes in any monotone path of the Boolean hypercube from $0^n$ to $1^n$. Monotone functions are exactly those functions that are $0$ on $0^n$ and have alternation number $1$. Alternation number has received attention in the past as a measure of the extent of non-monotonicity of a Boolean function. Markov \cite{M58} related it to the number of not-gates that any Boolean circuit realizing a Boolean function must use, which is yet another measure of non-monotonicity. Alternation number has also been studied in the context of various restricted circuit models \cite{SC93, ST04, M09a, M09b}, learning classes of circuits with bounded non-monotonicity \cite{BCOST15} and cryptography \cite{GMOR15}. Lin and Zhang \cite{LZ17} proved the \emph{Sensitivity Conjecture}\footnote{The Sensitivity Conjecture was later unconditionally proved by the seminal work of Huang; so it is not a conjecture any more.} and the \emph{Log-Rank Conjecture}, two important conjectures in complexity theory, for the special case of functions with constant alternation number. Krishnamoorthy and Sarma \cite{KS19} has studied relations of alternation number with various combinatorial, algebraic and analytic complexity measures of Boolean functions.

It follows from \Cref{prop:basic_facts} (4) and the fact that for monotone functions the measures certificate complexity and block sensitivity are equal \cite{BdeW02Survey} that for any monotone function $f$, $\dclass(f)=O(\bs(f)^2)$. The following theorem generalizes the  above result for the class of functions with bounded alternation. In addition, it also proves a quadratic upper bound on $\dclass(f)$ in terms of $\degree(f)$ for functions with constant alternation number. In particular, it gives a negative answer to Question~\ref{qn:q1} (1) and Question~\ref{qn:q1} (2) for functions with constant alternation number. 
\begin{theorem}
\label[theorem]{theorem:k-alternation-d-bs}
For every Boolean function $f:\{0,1\}^n\rightarrow\{0,1\}$,
\begin{enumerate}
\item $\dclass(f) = O(\alt(f)\cdot \bs(f)^2)$ and 
\item $\dclass(f) = O(\alt(f)\cdot \deg(f)^2)$.
\end{enumerate}
In particular, if $\alt(f)=O(1)$, then $\dclass(f) = O(\bs(f)^2)$ and $\dclass(f) = O(\deg(f)^2)$.

\end{theorem}
As stated in the preceding section, $\mathsf{TRIBES}_{\sqrt{n} \times \sqrt{n}}$ witnesses the tightness of \Cref{{theorem:k-alternation-d-bs}} (1).
\subsubsection{Decision trees and communication protocols}
 In our second category of results, we reduce the task of proving improved upper bounds on $\dclass(\cdot)$ for the class of all functions to potentially easier tasks of bounding smaller measures.
 
 Consider a bi-partition $(x^{(1)}, x^{(2)})$ of the bits of $x$ into two parts consisting of $n_1$ and $n_2$ bits respectively. In the two-party communication model introduced by Yao \cite{A79} there are two communicating parties Alice and Bob, holding $x^{(1)}$ and $x^{(2)}$ respectively. The parties wish to compute $f(x)$ by communicating bits. The objective is to accomplish this by communicating as few bits as possible. The \emph{communication complexity} of $f$, denoted by $\CC(f)$, is the number of bits that need to be exchanged in the worst case by any protocol that computes $f$. See Section~\ref{prelims} for more details.
 
 Suppose $f$ has a shallow decision tree $T$ of depth $d$. Then there also exists an efficient communication protocol that computes $f$, in which the parties simulate $T$ . Each query that $T$ makes can be answered by $1$ bit of communication. Thus, the protocol exchanges $d$ bits in the worst case. This gives us that $\CC(f)\leq \dclass(f)$.

We make the following observation. 
\begin{observation}
\label[observation]{obs:cc<D}\
\begin{enumerate}
    \item If $\dclass(f)= O(\bs(f)^2)$, then $\CC(f)=O(bs(f)^2)$, and 
    \item if $\dclass(f)= O(\degree(f)^2)$, then $\CC(f)=O(\degree(f)^2)$.
\end{enumerate}
\end{observation}
The question we ask is whether a converse of Observation~\ref{obs:cc<D} holds. We show that a ``approximate converse" of Observation~\ref{obs:cc<D} holds, up to a factor of $\mathsf{polylog}(n)$, when the antecedents and consequents are universally quantified.

\begin{theorem}
\label[theorem]{thm:cc-deg-bs-implications}\
\begin{enumerate}
    \item If for all Boolean functions $f$ and all bi-partitions of the input bits $\CC(f) = O(\deg(f)^2\cdot \mathsf{polylog}(n))$ then for all $f$, $\dclass(f) =O(\deg(f)^2\cdot \mathsf{polylog}(n))$. 
    \item If for all Boolean functions $f$ and all bi-partitions of the input bits $\CC(f) =O(\bs(f)^2\cdot \mathsf{polylog}(n))$ then for all $f$ $\dclass(f) =O(\fbs(f)^2\cdot \mathsf{polylog}(n))$. 
\end{enumerate}
\end{theorem}

\Cref{thm:cc-deg-bs-implications} opens up the possibility of deriving improved upper bounds on $\dclass(\cdot)$ by designing cheap communication protocols.\footnote{In fact, it follows from our proof that it is sufficient to design cheap protocols only for functions of the form $f \circ \ind_{c \log n+n^c}$ for some constant $c>0$, where $\ind_{m+2^m}$ is the indexing function on $m+2^m$ bits (see \Cref{def:ind}).} For instance, to prove Aaronson's conjecture (up to $\mathsf{polylog}(n)$) that the randomized query complexity is at most the fractional block sensitivity \cite{aaronson2008quantum} it is sufficient to show that the communication complexity is at most square of the block sensitivity times $\mathsf{polylog}(n)$. Thus, our task is reduced that of bounding $\CC(\cdot)$. The latter is a potentially easier task; there are Boolean functions $f$ for which $\CC(f)<<\dclass(f)$. For example, for the majority function $\mathsf{MAJ}_n$ on $n$ bits (see \Cref{def:maj}), $\CC(\mathsf{MAJ}_n) \leq \lceil \log n\rceil$ for any bi-partition of its input bits, but $\dclass(\mathsf{MAJ}_n)=n$. Furthermore, there are upper bounds on $\CC(f)$ by measures that are bounded above by $\dclass(f)$, and are far lower than $\dclass(f)$ for many $f$. One such measure is the \emph{decision tree rank} of $f$, denoted by $\rank(f)$ (See \Cref{def:rank}).

Decision tree rank was introduced by Ehrenfeucht and Haussler \cite{AD89} in the context of learning (also see \cite{ABDORU10}). More recently, Dahiya and Mahajan studied decision tree rank in the context of complexity measures of Boolean functions \cite{DM21}.

Decision tree rank is defined in terms of decision trees and is bounded above by $\dclass(\cdot)$. Furthermore, it can be arbitrarily smaller than $\dclass(\cdot)$. For example, for the AND function on $n$ bits, $\dclass(\cdot)=n$ and $\rank(\cdot)=1$.

We show (see \Cref{proposition:lifting}) that $\rank(\cdot)$ is an upper bound on $\CC(\cdot)$ up to a factor of $\log n$. We prove the proposition by relating rank to the \emph{subcube decision tree complexity}, (see \Cref{def:detsubcube}), which is the least depth of a decision tree that computes $f$ by querying indicator functions of subcubes at its internal nodes. Such a tree is efficiently simulable by a communication protocol.

There are powerful decision tree models that can be efficiently simulated by communication protocols; one such model is the subcube decision tree mentioned before. Another well-known model is the \emph{parity decision tree} (see \Cref{def:detparity}), which is a decision tree that computes $f$ by querying the parity of various subsets of bits of the input string at its internal nodes \cite{M09, THWX13, STV14}. A parity decision tree can be efficiently simulated by a communication protocol. Each parity query can be answered by an exchange of two bits. Thus $\CC(\cdot)=O(\pdt(\cdot))$. \Cref{thm:cc-deg-bs-implications} implies that bounds on the parity decision tree complexity (denoted by $\pdt(\cdot)$) of the class of all Boolean functions translate to comparable bounds on the decision tree complexity of all Boolean functions.

The following theorem is a direct consequence of \Cref{thm:cc-deg-bs-implications}, \Cref{proposition:lifting}, and the preceding discussion.
\begin{theorem}
\label[theorem]{thm:corollary_of_main_cc}
Let $\mathcal{M} \in \{\rank, \pdt\}$. Then
\begin{enumerate}
    \item If for all Boolean functions $f$,  $\mathcal{M}(f) = O(\deg(f)^2\cdot \mathsf{polylog}(n))$ then for all $f$, $\dclass(f) =O( \degree(f)^2\cdot \mathsf{polylog}(n))$.
    \item If for all Boolean functions $f$, $\mathcal{M}(f) = O(\bs(f)^2 \cdot \mathsf{polylog}(n))$ then for all $f$, $\dclass(f) =O( \fbs(f)^2\cdot \mathsf{polylog}(n))$. 
\end{enumerate}
\end{theorem}
\Cref{thm:cc-deg-bs-implications} is proved by a lifting theorem for communication complexity of composed functions where the inner function is the indexing function on $\Theta(\log n)$ bits (see \Cref{sec:proof-techniques} and \Cref{sec:communication}). The arity of the indexing gadget is the source of the $\mathsf{polylog}(n)$ factors in $\Cref{thm:cc-deg-bs-implications}$. A lifting theorem with a gadget of constant size is a major open problem in communication complexity, with many interesting consequences. It is apparent from our proof of \Cref{thm:cc-deg-bs-implications} that such a lifting theorem will remove the polylogarithmic factors. Our study thus strengthens the case for a lifting theorem with constant sized inner function.
\subsection{Proof Techniques}
\label[section]{sec:proof-techniques}
In this section we discuss at a high level the key ideas of our proofs.
\paragraph{Graph properties:} To prove \Cref{lemma:graph_key}, we assume, by complementing the property $\cP$ if necessary, that the empty graph is not in the property. Then we split the proof into two parts based on the smallest number of edges $m$ of any graph $G=(V,E)$ in the property.

If $m>n/4$, then we consider the restriction of $\cP$ where each $x_{\{u,v\}}$ such that $\{u,v\}\notin \cP$ is set to $0$. We observe that the restricted function is the logical AND of the unset variables, and hence has block sensitivity and degree both equal to $m=\Omega(n)$.

On the other hand, if $m\leq n/4$, then we show that by carefully restricting and identifying variables, it is possible to turn $\cP$ into a non-trivial symmetric function on $\Omega(n)$ variables which, by known results, has both degree and block sensitivity $\Omega(n)$ (see \Cref{prop:basic_facts} (9)). In this part of the proof, we crucially use the invariance of $\cP$ under permutation of vertices.
\paragraph{Zebra functions and functions with bounded alternation:}
In the proofs of \Cref{thm:zebra-d-bs-deg} (1) as well as \Cref{theorem:k-alternation-d-bs} (1) and (2), the following lemma plays a central role.
\begin{lemma}
\label[lemma]{lemma:d_n_cstar}
For any Boolean function $f:\{0,1\}^n\rightarrow \{0,1\}$, 
\begin{enumerate}
\item $\dclass(f)=O(\cminstar(f)\cdot\bs(f))$,
\item $\dclass(f)=O(\cminstar(f)\cdot\degree(f))$.
\end{enumerate}
\end{lemma}
\Cref{lemma:d_n_cstar} is proved by a natural adaptation of the proofs \Cref{prop:basic_facts} (3) and (4) \cite{BdeW02Survey, M04}. We provide a proof of \Cref{lemma:d_n_cstar} in Appendix~\ref{sec:cminstar_key}. In light of \Cref{lemma:d_n_cstar}, the proofs of the aforementioned results boil down to bounding the maximin certificate complexity of the functions in the respective classes. To prove \Cref{thm:zebra-d-bs-deg} (2), we bound the number of bits that need to be revealed to certify the stripe a specific input belongs to. To do that, we consider the monotone and anti-monotone functions obtained by treating the boundary of a specific stripe as a threshold, bound the certificate complexities of these functions, and show how those certificates can be put together to certify membership of a certain input in a specific stripe.
\paragraph{Communication complexity:} The proof of \Cref{thm:cc-deg-bs-implications} is based on a lifting theorem for deterministic communication complexity (see \Cref{{proposition:lifting-gpw}}), which essentially asserts that there for every function $g$, the deterministic communication complexity of $g \circ \ind_{c \log n + n^c}$ is asymptotically the same as $\log n \cdot \dclass(g)$, where $\ind_{m+2^m}$ is the indexing function on $m+2^m$ bits (see \Cref{def:ind}). The hypotheses of the two parts of \Cref{thm:cc-deg-bs-implications}, coupled with the lifting theorem, immediately implies upper bounds on $\dclass(g)$ in terms of the block sensitivity and degree of $g \circ \ind_{c\log n+n^c}$. We finish the proof by bounding the block sensitivity and degree of $g \circ \ind_{c\log n+n^c}$ in terms of the fractional block sensitivity and degree of $g$ respectively.
\subsection{Organization of the paper}
In \Cref{prelims} we lay down some notations that are used in this paper, and give some definitions. In \Cref{sec:graph} we  define graph properties, set some notations and prove \Cref{thm:graph_P}. In \Cref{sec:zebra} we formally define zebra functions and related concepts, state some basic properties of zebra functions, and prove \Cref{thm:zebra-d-bs-deg}. In \Cref{sec:k-alt} we prove the result about functions with bounded alternation number (\Cref{theorem:k-alternation-d-bs}). In \Cref{sec:communication} we prove our main result (\Cref{thm:cc-deg-bs-implications}) pertaining to connections of Question~\ref{qn:q1} to communication complexity and other decision tree models. The appendices contain missing proofs and a section on preliminaries.
\section{Definitions and Preliminaries}
\label[section]{prelims}
For any real number $t>0$, $\log t$ stands for the logarithm of $t$ to the base $2$. We denote the set of all integers by $\mathbf{Z}$. For a natural number $N$, $[N]$ denotes the set $\{1,\ldots,N\}$. Let $x=(x_1,\ldots, x_n), y=(y_1,\ldots,y_n) \in \{0,1\}^n$. We say that $x \leq y$ (resp. $x\geq y$) if for each $i \in [n]$, $x_i\leq y_i$ (resp. $x_i \geq y_i$). We say that $x<y$ (resp. $x>y$) if $x\leq y$ (resp. $x \geq y$) and $x \neq y$. $|x|$ denotes the \emph{Hamming weight} of $x$ defined as $\sum_{i=1}^n x_i$. A $1$-bit (resp. $0$-bit) of $x$ is an index $i\in[n]$ such that $x_i=1$ (resp. $x_i=0$). For $S \subseteq [n]$, $\{0,1\}^S$ denotes the set of all binary strings whose indices correspond to elements of $S$. For a set of indices $B \in [n]$, we denote the string obtained from $x$ by negating the bits in the locations with indices in $B$ by $x^{\oplus B}$. If $B=\{i\}$ is singleton, we abuse notation and write $x^{\oplus B}$ as $x^{\oplus i}$.
\begin{definition}[Monotone increasing path or monotone path, alternation number of monotone paths and functions]
Let $p, q \in \{0,1\}^n$ such that $p\leq q$. A \emph{monotone  increasing path} from $p$ to $q$ in the Boolean hypercube $\{0, 1\}^n$ is a sequence $(x^{(1)}, \ldots, x^{(k)})$ such that
\begin{itemize}
    \item Each element $x^{(i)}$ of the sequence is a bit string $(x^{(i)}_1, \ldots, x^{(i)}_n)\in \{0, 1\}^n$, for $i=1,\ldots, k$,
    \item $p=x^{(1)}$ and $q=x^{(k)}$,
    \item for each $i=1,\ldots, k-1$, there exists a $j \in [n]$ such that
    \begin{itemize}
        \item $x^{(i)}_j=0$ and $x^{(i+1)}_j=1$,
        \item $x^{(i)}_\ell=x^{(i+1)}_\ell$ for each $\ell \in [n] \setminus \{j\}$.
    \end{itemize}
\end{itemize}
We refer to such a path simply as a \emph{monotone path} for convenience.

Let $f:\{0,1\}^n \rightarrow \{0, 1\}$ be a Boolean function. The \emph{alternation number} of a monotone path $P=(x^{(1)}, \ldots, x^{(k)})$ with respect to $f$, denoted by $\alt_f(P)$, is the number of times $f$ changes value on the path. Formally,
\[\alt_f(P)=|\{i \in [k-1] \mid f(x^{(i)})\neq f(x^{(i+1)})\}|.\]
We often drop the subscript $f$ from the notation and denote alternation number of $P$ simply by $\alt(P)$ when the function is clear from the context.

The alternation number of $f$, denoted by $\alt(f)$, is defined to be the maximum alternation number of any monotone path from $0^n$ to $1^n$ with respect to $f$. 
\end{definition}
\begin{observation}
\label[observation]{obs:path}
Let $x,y \in \{0,1\}^n$ and $x \leq y$. Then there exists a monotone path from $x$ to $y$.
\end{observation}
\begin{observation}
\label[observation]{obs:alt_restr}
The alternation number of a function $f$ is at least the alternation number of any restriction of $f$. That is, for every integer $k\geq 0$ and every subcube $C$, if $\alt(f) \leq k$ then $\alt(f\mid_C) \leq k$.
\end{observation}


\section{Graph Properties}
\label[section]{sec:graph}
In this section, we first define graph properties and lay down few notations. We then proceed to prove \Cref{lemma:graph_key}, which immediately implies \Cref{thm:graph_P}.

Let $\G_n$ be the set of all simple undirected graphs on $n$-vertices. A graph property is a function $\cP: \G_n\rightarrow\{0,1\}$ which is invariant under permutation of vertices. In other words, if two graphs are isomorphic, then either both belong to the property, or both do not.

Example of a graph property is one that maps connected graphs to $1$ and disconnected graphs to $0$. But, the Boolean function that maps exactly those graphs that have an edge between vertices $1$ and $2$ to $1$, is not a graph property, as it is not invariant under relabling of vertices.

A non-constant graph property is also called \emph{non-trivial}. A graph $G \in \G_n$ is often identified with a string $x_G \in \{0,1\}^{\binom{n}{2}}$. Each location of $x_G$ corresponds to a distinct unordered pair of distinct vertices of $G$. $\left(x_G\right)_{\{u,v\}}:=1$ if there is an edge in $G$ between vertices $u$ and $v$, and $\left(x_G\right)_{\{u,v\}}:=0$ otherwise. We use $\cP(G)$ and $\cP(x_G)$ interchangeably.
\begin{proof}[Proof of \Cref{lemma:graph_key}]
Wlog. we assume that the empty graph is not in $\cP$. Let $m>0$ be the least number of edges in any graph in $\cP$. Let $G=(V,E)\in\cP$ have $m$ edges. We consider the following cases. \\

\noindent\textsf{Case 1} $m> \frac{n}{4}:$ Restrict $\cP$ by setting $\left(x_G\right)_{\{u,v\}}$ to $0$ for each $\{u,v\} \notin E$. By the minimality of $G$, the restriction is the $\textsf{AND}$ function on the $m$ variables corresponding to the edges in $G$. Thus the block sensitivity and degree of the restriction are both at least $n/4$. Since degree and block sensitivity do not increase under restriction (\Cref{prop:basic_facts} (7)), therefore $\bs(\cP)> n/4$ and $\degree(\cP)> n/4$.  \\

\noindent\textsf{Case 2} $m\leq  \frac{n}{4}:$ In the graph $G$, there are at least $\frac{n}{2}$ vertices that have degree zero. Let the vertices with non-zero degrees be $v_1, v_2, \cdots v_k$, where $k\leq \frac{n}{2}$. Let the remaining vertices be $v_{k+1}, \cdots v_{n}$.
Without loss of generality let $v_k$ be connected to $v_1, \cdots, v_d$. Let $G'$ be the graph obtained from $G$ by removing all edges incident on $v_k$.  Since $G'$ has less than $m$ edges, $G'\notin\cP$.  Now for each $j\in \{k, k+1, \cdots, n\}$, consider the graph $G_j$ obtained by adding the edges $\{j,v_1\}, \{j,v_2), \cdots \{j,v_d\}$ to $G'$. All the $G_j$s formed this way are isomorphic to $G$, and are thus all in $\cP$. Also note that $G_k$ is the same as $G$.
 
 Now consider a function $f:\{0, 1\}^{n-k+1}\rightarrow \{0,1\}$ defined as follows. Let $x=(x_k, x_{k+1},\ldots,x_n)\in\{0, 1\}^{n-k+1}$. We will now define a graph $G_x$ on $n$ vertices, and then define $f(x):=\cP(G_x)$. $G_x$ contains exactly the following edges.
 \begin{itemize}
     \item All edges of $G'$.
     \item Edges $\{v_i, v_1\}, \ldots, \{v_i, v_d\}$ for each $i \in \{k, k+1,\ldots,n\}$ such that $x_i=1$.
 \end{itemize}
 First, note that $f(0^{n-k+1})=0$ (since $G_{0^{n-k+1}}=G' \notin \cP$) and for any string $x$ with $|x|=1$, $f(x)=1$ (since in this case $G_x$ is the same as $G_j$ for some $j \in \{k,\ldots,n\}$, and is hence in $\cP$). Thus, $f$ is a non-constant function.
 
 Next, note that for any $x,x'\in \{0, 1\}^{n-k+1}$ such that $|x|=|x'|$, $G_x$ and $G_{x'}$ are isomorphic, and hence $f(x)=\cP(G_x)=\cP(G_{x'})=f(x')$. Thus, $f$ is a symmetric function. 
 
 Since any non-constant symmetric function has linear degree and block sensitivity (\Cref{prop:basic_facts} (8)), therefore $\degree(f)$ and $\bs(f)$ are both $\Omega(n-k+1)=\Omega(n)$ (since $k\leq n/2$).
 
 Finally, $f$ is obtained from $\cP$ by restriction and identification of variables, as follows. First, every variable corresponding to a pair $\{v_i, v_j\}$ such that $1 \leq i <j \leq k-1$ and $\{i,j\}$ is not an edge of $G'$, is set to $0$. Next, for each $i\in \{k, k+1, \ldots, n\}$, the variables corresponding to pairs $\{v_i, v_1\},\ldots,\{v_i, v_d\}$ are identified.
 
 Since degree and block sensitivity do not increase under restriction and identification of variables (\Cref{prop:basic_facts} parts (7) and (8)), therefore, $\degree(\cP)$ and $\bs(\cP)$ are both $\Omega(n)$.
\end{proof}
\section{Zebra functions}
\label[section]{sec:zebra}
In this section, we first formally define zebra functions and some related concepts, and state some basic facts about zebra functions. We then proceed to prove \Cref{thm:zebra-d-bs-deg}. 
\subsection{Definitions and basic concepts}
\begin{definition}[Zebra function]
A function $f:\{0,1\}^n \rightarrow \{0, 1\}$ is called a \emph{zebra function} if all monotone paths from $0^n$ to $1^n$ have same alternation number.
\end{definition}
Visually, a zebra function induces a partition of the Boolean hypercube into \emph{stripes}. In each stripe, the value of the function is constant (i.e., each stripe is monochromatic with respect to the function). Furthermore, the value of the function is different in adjacent stripes. Each monotone path from $0^n$ to $1^n$ must pass through each stripe (See \Cref{fig:zebra}). We now build up towards defining the stripes formally.
\begin{proposition}
\label[proposition]{prop:groundwork-for-stripes}
Let $f$ be a zebra function on $n$ bits and $x, y \in \{0, 1\}^n$ be such that $x \leq y$. Let $P_1$ and $P_2$ be two monotone paths from $x$ to $y$. Then, $\alt_f(P_1)=\alt_f(P_2)$.
\end{proposition}
\begin{proof}
Towards a contradiction, assume that $\alt_f(P_1)\neq\alt_f(P_2)$. Consider any monotone paths $Q_1$ and $Q_2$ from $0^n$ to $x$ and from $y$ to $1^n$ respectively. Let $R_1$ (resp. $R_2$) be the the monotone path from $0^n$ to $1^n$ obtained by concatenating $Q_1, P_1$ and $Q_2$ (resp. $Q_1, P_2$ and $Q_2$). Then, $\alt_f(R_1)=\alt_f(Q_1)+\alt_f(P_1)+\alt_f(Q_2)$ and $\alt_f(R_2)=\alt_f(Q_1)+\alt_f(P_2)+\alt_f(Q_2)$. Thus, $\alt_f(R_1)\neq \alt_f(R_2)$, contradicting the hypothesis that $f$ is a zebra function.
\end{proof}
In light of \Cref{prop:groundwork-for-stripes} we extend the definition of alternation number to points with respect to zebra functions.
\begin{definition}[Alternation number of a point with respect to a zebra function]
\label[Definition]{alt-of-points}
Let $f$ be a zebra function on $n$ bits and $x \in \{0, 1\}^n$. The alternation number of $x$ with respect to $f$, denoted by $\alt_f(x)$, is defined as $\alt_f(P)$ for any monotone path from $0^n$ to $x$. \Cref{prop:groundwork-for-stripes} guarantees that $\alt_f(x)$ is well-defined.
\end{definition}
Now we are ready to define stripes of a zebra function. Let $f$ be a zebra function on $n$ bits and $\alt(f)=k$. For $i=0, \ldots, k$, the $i$-th stripe of $f$, $\stripe_f(i)$, is defined to be the set $\{x \in \{0, 1\}^n \mid \alt_f(x)=i\}$.  An easy induction on $i$ establishes that $\stripe_f(i)$ is monochromatic with respect to $f$. For $i=1,\ldots,k$, we denote by $f(i)$ the value of $f$ on any point in $\stripe_f(i)$.

$x\in\stripe_f(i)$ is called a \emph{minimal point} of $\stripe_f(i)$ if for each point $y<x$, $y \notin \stripe_f(i)$. Similarly, $x\in\stripe_f(i)$ is called a \emph{maximal point} of $\stripe_f(i)$ if for each point $y>x$, $y \notin \stripe_f(i)$.

The following proposition lists some useful facts about zebra functions; a proof may be found in \Cref{sec:zebra-facts}.
\begin{proposition}
\label[proposition]{prop:zebra-useful-facts}
Let $f$ be a zebra function on $n$ bits and $\alt(f)=k$. Then the following hold.
\begin{enumerate}
    \item For $0\leq i<j\leq k$, $\stripe_f(i) \cap \stripe_f(j)=\emptyset$. Also, $\cup_{i=1}^k \stripe_f(i)=\{0, 1\}^n$. In other words, the stripes form a partition of the Boolean hypercube.
    \item For $i=0,\ldots,k-1$, $f(i)\neq f(i+1)$.
    \item Let $i\in \{0,\ldots,k\}$ and $x \in \stripe_f(i)$. Then there exists a $y \leq x$ (resp. $y \geq x$) such that $y$ is a minimal (resp. maximal) point of $\stripe_f(i)$.
    \item Let $x=(x_1,\ldots,x_n)$ be a minimal (resp. maximal) point of $\stripe_f(i)$. Let $x_i=1$ (resp. $x_i=0$). Let $y$ be a point obtained from $x$ by changing the value of $x_i$ and leaving other bits of $x$ unchanged. Then $x\in\stripe_f(i-1)$ (resp. $x \in \stripe_f(i+1)$). In particular, $f(y) \neq f(x)$.
    \item For any $j\in[n]$ and $b\in\{0,1\}$, the function $f\mid_{x_j=b}$ is a zebra function. In other words, the class of all zebra functions is closed under restrictions of variables.
\end{enumerate}
\end{proposition}
\subsection{Proof of \Cref{thm:zebra-d-bs-deg}}
We need the following lemma.
\begin{lemma}
\label[lemma]{lemma:C-of-f_i}
Let $f: \{0, 1\}^n \rightarrow \{0, 1\}$ be a zebra function and $\alt(f)=k$. For $i=0,\ldots,k$, define
\[f_i(x)=\left\{\begin{array}{cc}1 & \mbox{if $x\in \cup_{j \geq i} \stripe_f(j)$,}\\0 & \mbox{otherwise.}\end{array}\right.\]
Then, $\cer(f_i)=O(\degree(f)^2)$.
\end{lemma}
\begin{proof}
The Lemma is trivial for $i=0$. So, we assume that $1 \leq i \leq k$. We show that $\cer_1(f_i)=O(\degree(f)^2)$. The argument that $\cer_0(f_i)=O(\degree(f)^2)$ is analogous. First, note that $f_i$ is a zebra (monotone) function with stripes $\stripe_{f_i}(0)=\cup_{j<i}\stripe_f(j)$ and $\stripe_{f_i}(1)=\cup_{j\geq i}\stripe_f(j)$. Next, fix an input $x$ such that $f_i(x)=1$. Thus $x\in \stripe_f(j)$ for some $j\geq i$. Let $y\leq x$ be a minimal input of $\stripe_f(i)$ (\Cref{prop:zebra-useful-facts} (3)). $y$ is also a minimal input of $\stripe_{f_i}(1)$. Thus $f_i(y)=1$. Since $f_i$ is monotone the $1$-bits of $y$ form a certificate of $x$ with respect to $f_i$. By \Cref{prop:zebra-useful-facts} (4), all the $1$-bits of $y$ are sensitive with respect to $f$. This proves that $\cer_1(f_i, x)\leq s(f) \leq \degree(f)^2$ (\Cref{prop:basic_facts} (3)). Since $x$ is an arbitrary $1$-input, the claim follows.
\end{proof}
\begin{proof}[Proof of \Cref{thm:zebra-d-bs-deg}]
Let $f: \{0, 1\}^n \rightarrow \{0, 1\}$ be a non-constant zebra function. \\
(Part 1) By \Cref{lemma:d_n_cstar} it is sufficient to prove that $\cminstar(f) \leq bs(f)$. Let $\alt(f)=k\geq 1$, and let $x=(x_1,\ldots,x_n) \in \{0, 1\}^n$ be a minimal point in $\stripe_f(k)$ (\Cref{prop:zebra-useful-facts}(3)). Let $t$ be the number of $1$-bits in $x$. By \Cref{prop:zebra-useful-facts}(4), every $1$-bit in $x$ is sensitive. Thus, $\bs(f) \geq \s(f)\geq t$. On the other hand, every input $(y_1,\ldots,y_n)$ such that $y_i=1$ whenever $x_i=1$, lies in $\stripe_f(k)$, and hence $f(y)=f(x)$; thus the $1$-bits of $x$ form a certifcate of $x$. Hence $\cmin(f)\leq \cer(f, x)\leq t\leq \bs(f)$. Now since the class of zebra functions is closed under restriction of variables (\Cref{prop:zebra-useful-facts}(5)), and block sensitivity does not grow under restriction of variables, therefore for any restriction $f'$ of $f$ to a subcube, $\cmin(f')\leq\bs(f')\leq\bs(f)$. Since, this holds for every restriction $f'$ of $f$, we conclude that $\cminstar(f)\leq\bs(f)$, and the proof is complete.

(Part 2) Let $x \in \stripe_f(i)$. Let the functions $f_i$ be as defined in \Cref{lemma:C-of-f_i}. By \Cref{lemma:C-of-f_i}, $x$ has a $1$-certificate $c_1$ of size $O(\degree(f))^2$ with respect to $f_i$. If $i=k$, that also certifies that $x \in \stripe_f(k)$. Since, $\stripe_f(k)$ is monochromatic with respect to $f$, we are through. If $i<k$, by \Cref{lemma:C-of-f_i} $x$ has a $0$-certificate $c_0$ of size $O(\degree(f))^2$ with respect to $f_{i+1}$. The concatenation of $c_1$ and $c_0$ certifies that $x \in \stripe_f(i)$. Since, $\stripe_f(i)$ is monochromatic with respect to $f$, the theorem follows.
\end{proof}
\section{Functions with bounded alternation number}
\label[section]{sec:k-alt}
In this Section, we prove \Cref{theorem:k-alternation-d-bs}.
\begin{proof}[Proof of \Cref{theorem:k-alternation-d-bs}]
Let $\alt(f)\leq k$. We will show that $\cmin(f) \leq k\cdot\bs(f)$ and $\cmin(f) \leq k\cdot\degree(f)$. Since, alternation number and block sensitivity does not increase unser restrictions (Observation~\ref{obs:alt_restr} and \Cref{prop:basic_facts} (7)), therefore we will have that $\cminstar(f) \leq k\cdot\bs(f)$ and $\cminstar(f) \leq k\cdot\degree(f)$. By \Cref{lemma:d_n_cstar} the theorem will follow.

(Part (i): $\cmin(f) \leq k\cdot\bs(f)$) Consider the subcube returned by \Cref{algo: iteration generation}.\\
\begin{algorithm}[H]
\label[algorithm]{algo: iteration generation}
\SetAlgoLined
\textbf{Initialize } $S^{(0)}\gets \emptyset, C^{(0)}\gets\{0, 1\}^n, t \gets 1$, $x^{(0)}\gets 0^n$. \\
\While{$f\mid_{C^{(t-1)}}$ is not a constant function}{
    Let $x^{(t)}>x^{(t-1)}$ be a point with minimum Hamming weight such that $f(x^{(t)})\neq f(x^{(t-1)})$.\\
    $S^{(t)}\gets \{i\in [n] \mid x^{(t)}_i=1\}$.\\
    $C^{(t)}\gets \{z \in \{0,1\}^n \mid z_i=1\mbox{ for all }i\in S^{(t)}\}$. \\
    $t\gets t+1$. \\
}
Return $C^{(t-1)}$.
\caption{}
\end{algorithm}
We first show that the step 3 of \Cref{algo: iteration generation} is well-defined: there exists an $x^{(t)}$ for the algorithm to choose. Since $f$ is non-constant on $C^{(t-1)}$ and $x^{(t-1)}\in C^{(t-1)}$, therefore there exists another point $x^{(t)} \in C^{(t-1)}$ such that $f(x^{(t)})\neq f(x^{(t-1)})$. Next, note that $x^{(t-1)}_i=1$ for all $i\in S^{(t-1)}$ and $x^{(t-1)}_i=0$ otherwise. Thus, for any other point $x^{(t)}\in C^{(t-1)}$, $x^{(t)}>x^{(t-1)}$.
Thus, an $x^{(t)}$ is guaranteed to exist for the algorithm to pick.

Let $\ell$ be the number of iterations of \Cref{algo: iteration generation}. Thus, $C^{(\ell)}$ is the subcube returned by \Cref{algo: iteration generation}. Clearly, $f\mid_{C^{(\ell)}}$ is a constant function; thus \Cref{algo: iteration generation} indeed returns a certificate with co-dimension $\codim(C^{(\ell)})=|S^{(\ell)}|$. Now by the choice of $x^{(t)}$, any index $j$ such that $x^{(t)}_j=1$ and $x^{(t-1)}_j=0$ is sensitive for $x^{(t)}$. Thus the number of such locations is at most $s(f)$. Thus for each $t=1,\ldots,\ell$, $|S^{(t)}|\leq|S^{(t-1)}|+\s(f)$, giving us $\codim(C^{(\ell)})=|S^{(\ell)}|\leq \ell\cdot\s(f) \leq \ell\cdot\bs(f)$ (\Cref{prop:basic_facts} (1)).

Now, it follows from Observation~\ref{obs:path} that there exists a monotone path from $0^n$ to $x^{(\ell)}$ that passes through each $x^{(t)}$ for $t=1,\ldots,\ell-1$. The alternation number of the path is at least $\ell$, giving us $\ell\leq k$. Hence $\codim(C^{(\ell)})\leq  k\cdot\bs(f)$.

(Part (ii): $\cmin(f) \leq k\cdot\degree(f)$) Consider the subcube returned by~\Cref{algo: iteration generation_poly}.\\
\begin{algorithm}[H]
\label[algorithm]{algo: iteration generation_poly}
\SetAlgoLined
\textbf{Initialize } $S^{(0)}\gets \emptyset, C^{(0)}\gets\{0, 1\}^n, t \gets 1$, $x^{(0)}\gets 0^n$. \\
\While{$f\mid_{C^{(t-1)}}$ is not a constant function}{
    Let $M$ be a maximal monomial of $P_{f\mid_{C^{(t-1)}}}$.\\
    Let $x^{(t)}>x^{(t-1)}$ be a point such that $x^{(t)}_i=0$ for all $i\notin M \cup S^{(t-1)}$ and $f(x^{(t)})\neq f(x^{(t-1)})$.\\
    $S^{(t)}\gets \{i\in [n] \mid x^{(t)}_i=1\}$.\\
    $C^{(t)}\gets \{z \in \{0,1\}^n \mid z_i=1\mbox{ for all }i\in S^{(t)}\}$.\\
    $t\gets t+1$.\\
    }
Return $C^{(t-1)}$.
\caption{}
\end{algorithm}
First, we show that step 4 of \Cref{algo: iteration generation_poly} is well defined: there indeed exists an $x^{(t)}$ for the algorithm to choose. $P_{f\mid_{C^{(t-1)}}}$ is a polynomial on variables $\{x_i \mid i\notin S^{(t-1)}\}$ obtained by substituting $1$ for  the variables $x_i$ for all $i\in S^{(t-1)}$. Now, substitute $0$ for all variables outside of $M$ in $P_{f\mid_{C^{(t-1)}}}$. Since $M$ is a maximal monomial of $P_{f\mid_{C^{(t-1)}}}$, we are left with a non-zero polynomial $P'$, say, on the remaining set of variables $\{x_i \mid i \in M \setminus S^{(t-1)}\}$. Since $P'$ is non-constant, therefore there exists a non-zero input $z=(z_{i})_{i\in M\setminus S^{(t-1)}}$ such that \begin{align}
\label[equation]{eqn:eq2}
    P'(z) \neq P'(0^{M\setminus S^{(t-1)}})=P_{f\mid_{C^{(t-1)}}}(0^{[n]\setminus S^{(t-1)}})=P_f(x^{(t-1)})=f(x^{(t-1)}).
\end{align}
    Now, define $x^{(t)}$ as follows:
\[x^{(t)}_i=\left\{\begin{array}{ccc}z_i & \mbox{if $i \in M\setminus S^{(t-1)}$,} \\ 1 & \mbox{if $i \in S^{(t-1)}$,}\\$0$ & \mbox{otherwise.}\end{array}\right.\]
Thus $P'(z)=P_f(x^{(t)})=f(x^{(t)})$. Furthermore, since $x^{(t-1)}_i=1$ if $i\in S^{(t-1)}$ and $x^{(t-1)}_i=0$ otherwise, we have that $x^{(t)}>x^{(t-1)}$. \Cref{eqn:eq2} now lets us conclude that $x^{(t)}$ satisfies the criteria to be chosen by the algorithm in step $4$.

Now we analyze the subcube output by \Cref{algo: iteration generation_poly}. Let $\ell$ be the number of iterations of \Cref{algo: iteration generation_poly}. Hence $C^{(\ell)}$ is the subcube returned by \Cref{algo: iteration generation_poly}. Clearly, $f\mid_{C^{(\ell)}}$ is a constant function; thus \Cref{algo: iteration generation_poly} indeed returns a certificate with co-dimension $\codim(C^{(\ell)})=|S^{(\ell)}|$. Now by the choice of $x^{(t)}$, $x^{(t)}$ has at most $|M| \leq \degree(f)$ $1$s more than $x^{(t-1)}$. Thus for each $t=1,\ldots,\ell$, $|S^{(t)}|\leq |S^{(t-1)}|+\degree(f)$, giving us $\codim(C^{(\ell)})=|S^{(\ell)}|\leq \ell\cdot\degree(f)$. Similar to Part (i),  it follows that $\ell\leq k$, concluding the proof.
\end{proof}

\section{Connections to communication complexity and other measures}
\label[section]{sec:communication}
In this section we prove Theorem~\ref{thm:cc-deg-bs-implications} and \Cref{thm:corollary_of_main_cc}. First we state few results needed to prove the theorem.
\begin{proposition}\label[proposition]{proposition:lifting}\
\begin{enumerate}
    \item \label[proposition]{proposition:lifting1} For every Boolean function $f: \{0,1\}^n \rightarrow \{0,1\}$ and for any arbitrary bi-partition of the indices into two disjoint sets $S_1, S_2 \subseteq [n]$ such that $S_1 \cap S_2 = \emptyset$ and $S_1 \cup S_2 = [n]$,  
    \[\CC(f) \leq 2 \cdot \subcubedt(f)\leq 2 \cdot \dclass(f).\]
    
     \item \label[proposition]{proposition:lifting2}For every Boolean function $f$, $\rank(f) \leq \subcubedt(f) \leq \rank(f) \cdot (\log n+1)$.
     
\end{enumerate}
\end{proposition}
Proof of this proposition is deferred to Appendix~\ref{section-proof-of-lifting}.

Let $\ind_{m+2^m}$ denote the indexing function on $m+2^m$ bits (see \Cref{def:ind}). G{\"o}{\"o}s, Pitassi and Watson proved the following lifting theorem by using a simulation method by Raz and Mckenzie, which essentially asserts that the communication protocol for $f \circ \ind_{c \log n + n^c}$ that simulates an optimal decision tree of $f$, serving each query by solving the corresponding copy of the indexing function, is asymptotically optimal.
\begin{theorem}{\cite{GPW18,raz1997separation}}\label[theorem]{proposition:lifting-gpw}
There exists a constant $c>0$ such that for every Boolean function $f:\{0,1\}^n\rightarrow \{0,1\}$,
\[\CC(f \circ \ind_{c\log n + n^c}) = \Theta(\log n \cdot D(f)),\]
where the communication complexity is with respect to the bi-partition of the input variables where the address bits of the input string of each copy of indexing are with Alice, and the corresponding target bits are with Bob.
\end{theorem}

The following proposition gives an exact relation between the degree of composition of two functions and the degrees of the individual functions.
\begin{proposition}\label{proposition:tal-deg-compositon1}\label{proposition:tal-bs-compositon} \cite{tal2013properties}
Let $f:\{0, 1\}^n \rightarrow \{0, 1\}$ and $g:\{0, 1\}^m \rightarrow \{0,1\}$
be two Boolean functions, then $\degree(f\circ g) = \degree(f)\cdot \degree(g)$
\end{proposition}
\begin{proof}{of Theorem~\ref{thm:cc-deg-bs-implications}:}
(Part 1) Assume the hypothesis of the statement of the theorem and let $f:\{0,1\}^n\rightarrow \{0,1\}$ be an arbitrary Boolean function. Let $c$ be the constant from \Cref{proposition:lifting-gpw}, and $m=c\log n$. Consider the function $F= f\circ \ind_{m+2^m}$. Consider the bi-partition of the input bits of $F$, in which Alice gets the address bits and Bob gets the target bits of the inputs to each copy of the inner indexing function. Invoking the hypothesis of the theorem on $F$ with respect to this bi-partition, we have that $\CC(F)= O(\degree(F)^2 \cdot \mathsf{polylog}(n))$. Combining this with Theorem~\ref{proposition:lifting-gpw} we get that $\dclass(f)= O(\CC(F)/ m) = O(\degree(F)^2\cdot \mathsf{polylog}(n))$.
Then using Proposition~\ref{proposition:tal-deg-compositon1} and by noting that $\degree(\ind)=m+1$ we get $\dclass(f) =O(\degree(f)^2\cdot \mathsf{polylog}(n))$.

(Part 2) Assume the hypothesis of the statement of the theorem, let $f:\{0,1\}^n\rightarrow \{0,1\}$ be any arbitrary Boolean function, and let $F$ be defined as in part 1. Invoking the hypothesis of the theorem on $F$ with respect to the bi-partition of part 1, we have that $\CC(F)= O(\bs(F)^2 \cdot \mathsf{polylog}(n))$. Combining this with Theorem~\ref{proposition:lifting-gpw} we get that $\dclass(f)= O(\CC(F)/ m) = O(\bs(F)^2\cdot \mathsf{polylog}(n))$. We finish by showing that $\bs(F) \leq (m+1) \cdot \fbs(f)$, which is what the rest of this proof focusses on.

Let $\bs(F)=k$. We will show that $\fbs(f)\geq \frac{k}{m+1}$. Let $x:=(x^{(1)},x^{(2)}, \cdots, x^{(n)})$ be an input such that $\bs(F, x)=k$. Let $\mathcal{B}=\{B_1, B_2, \cdots, B_k\}$ be a set of $k$ minimal disjoint sensitive blocks of $x$. 

We claim that for $i\in [k]$, at most $(m+1)$ blocks from $\mathcal{B}$ intersect $x^{(i)}$. To establish the claim, we start by observing that to change the output of indexing function on $x^{(i)}$ it is necessary to either change an address bit of $x^{(i)}$, or the unique target bit pointed to by the address bits. It follows from the minimality of the blocks in $\mathcal{B}$ that if a $B_j$ overlaps nontrivially with a $x^{(i)}$, then $B_j$ must contain either an address bit of $x^{(i)}$, or the unique target bit in $x^{(i)}$ that the address bits point to. Since there are $m$ address bits, one target bit, and the blocks in $\mathcal{B}$ are disjoint, the claim follows.

Define $y_i:=\ind_{m+2^m}(x^{(i)})$ and define $y=(y_1,\ldots, y_n)$. Note that $y$ is in the domain of $f$. For each block $B_j \in \mathcal{B}$ define $A_j = \{\ell \in [n] \mid x^{(\ell)}\cap B_j \neq \emptyset\}$. By the minimality of the blocks in $\mathcal{B}$ each $A_j$ is a sensitive block of $y$ with respect to $f$.

Consider the following assignment of weights to the sensitive (with respect to $f$) blocks of $y$. Each sensitive block $A$ of $y$ is assigned weight $w_A:=\frac{|\{j \in [k] \mid A_j=A\}|}{m+1}$. We first show that this assignment is a feasible point of the fractional block sensitivity LP. Fix an $i \in [n]$. Now, each sensitive block of $y$ of positive weight that contains $i$ is in fact $A_j$ for some $j\in[k]$ such that $B_j \cap x^{(i)} \neq \emptyset$. Now, since we have shown that the number of such $j$'s is at most $m+1$, it follows that the sum of the weights of all sensitive blocks of $y$ containing $i$ is at most $1$. Thus the assignment is feasible for the $\fbs$ LP. Hence,
\begin{align*}
\fbs(f) &\geq \sum_{A} w_A \\
& = \sum_A\frac{|\{j \in [k] \mid A_j=A\}|}{m+1} \\
&=\frac{k}{m+1}
\end{align*}
where the all the sums are  over all the sensitive blocks of $y$.
\end{proof}
We now prove Theorem~\ref{thm:corollary_of_main_cc} which is an immediate corollary of Theorem~\ref{thm:cc-deg-bs-implications}.
\begin{proof}{of Theorem~\ref{thm:corollary_of_main_cc}:}
Assume the hypothesis of the theorem. By Proposition~\ref{proposition:lifting} (\ref{proposition:lifting1}) and (\ref{proposition:lifting2}) for any Boolean function $f$, and for any bi-partition of $f$, $\CC(f) =O(\subcubedt(f)\cdot \mathsf{polylog}(n)) = O(\rank(f)\cdot \mathsf{polylog}(n))$. It is easy to see that a communication protocol can answer each query of a parity decision tree by exchanging 2 bits. Thus, $\CC(f) = O(\pdt(f))$. The theorem follows from  Theorem~\ref{thm:cc-deg-bs-implications}.

\end{proof}

\bibliographystyle{alpha}

\begin{thebibliography}{ABDK{\etalchar{+}}21}

\bibitem[Aar08]{aaronson2008quantum}
Scott Aaronson.
\newblock Quantum certificate complexity.
\newblock {\em Journal of Computer and System Sciences}, 74(3):313--322, 2008.

\bibitem[ABD{\etalchar{+}}10]{ABDORU10}
James Aspnes, Eric Blais, Murat Demirbas, Ryan O’Donnell, Atri Rudra, and
  Steve Uurtamo.
\newblock k+ decision trees.
\newblock In {\em International Symposium on Algorithms and Experiments for
  Sensor Systems, Wireless Networks and Distributed Robotics}, pages 74--88.
  Springer, 2010.

\bibitem[ABDK{\etalchar{+}}21]{ABKRT21}
Scott Aaronson, Shalev Ben-David, Robin Kothari, Shravas Rao, and Avishay Tal.
\newblock Degree vs. approximate degree and quantum implications of huang’s
  sensitivity theorem.
\newblock In {\em Proceedings of the 53rd Annual ACM SIGACT Symposium on Theory
  of Computing}, pages 1330--1342, 2021.

\bibitem[Ama11]{A11}
Kazuyuki Amano.
\newblock Minterm-transitive functions with asymptotically smallest block
  sensitivity.
\newblock {\em Information processing letters}, 111(23-24):1081--1084, 2011.

\bibitem[And79]{A79}
C-C~Yao Andrew.
\newblock Some complexity questions related to distributed computing.
\newblock In {\em Proc. 11th STOC}, pages 209--213, 1979.

\bibitem[BBC{\etalchar{+}}01]{BBCMdW01}
Robert Beals, Harry Buhrman, Richard Cleve, Michele Mosca, and Ronald de~Wolf.
\newblock Quantum lower bounds by polynomials.
\newblock {\em Journal of the ACM (JACM)}, 48(4):778--797, 2001.

\bibitem[BCO{\etalchar{+}}15]{BCOST15}
Eric Blais, Cl{\'{e}}ment~L. Canonne, Igor~Carboni Oliveira, Rocco~A. Servedio,
  and Li{-}Yang Tan.
\newblock Learning circuits with few negations.
\newblock In Naveen Garg, Klaus Jansen, Anup Rao, and Jos{\'{e}} D.~P. Rolim,
  editors, {\em Approximation, Randomization, and Combinatorial Optimization.
  Algorithms and Techniques, {APPROX/RANDOM} 2015, August 24-26, 2015,
  Princeton, NJ, {USA}}, volume~40 of {\em LIPIcs}, pages 512--527. Schloss
  Dagstuhl - Leibniz-Zentrum f{\"{u}}r Informatik, 2015.

\bibitem[BDW02]{BdeW02Survey}
Harry Buhrman and Ronald De~Wolf.
\newblock Complexity measures and decision tree complexity: a survey.
\newblock {\em Theoretical Computer Science}, 288(1):21--43, 2002.

\bibitem[BvEBL74]{BBL74}
M.R. Best, P.~van Emde~Boas, and H.W. Lenstra.
\newblock A sharpened version of the aanderaa-rosenberg conjecture.
\newblock {\em Stichting Mathematisch Centrum. Zuivere Wiskunde. Stichting
  Mathematisch Centrum.}, 1974.

\bibitem[CK07]{CK07}
Amit Chakrabarti and Subhash Khot.
\newblock Improved lower bounds on the randomized complexity of graph
  properties.
\newblock {\em Random Struct. Algorithms}, 30(3):427--440, 2007.

\bibitem[DK99]{DK99}
Yevgeniy Dodis and Sanjeev Khanna.
\newblock Space-time tradeoffs for graph properties.
\newblock In {\em International Colloquium on Automata, Languages, and
  Programming}, pages 291--300. Springer, 1999.

\bibitem[DM21]{DM21}
Yogesh Dahiya and Meena Mahajan.
\newblock On (simple) decision tree rank.
\newblock In {\em 41st IARCS Annual Conference on Foundations of Software
  Technology and Theoretical Computer Science (FSTTCS 2021)}. Schloss
  Dagstuhl-Leibniz-Zentrum f{\"u}r Informatik, 2021.

\bibitem[DS19]{KS19}
Krishnamoorthy Dinesh and Jayalal Sarma.
\newblock Alternation, sparsity and sensitivity: Bounds and exponential gaps.
\newblock {\em Theoretical Computer Science}, 771:71--82, 2019.

\bibitem[EH89]{AD89}
Andrzej Ehrenfeucht and David Haussler.
\newblock Learning decision trees from random examples.
\newblock {\em Information and Computation}, 82(3):231--246, 1989.

\bibitem[GMOR15]{GMOR15}
Siyao Guo, Tal Malkin, Igor~C Oliveira, and Alon Rosen.
\newblock The power of negations in cryptography.
\newblock In {\em Theory of Cryptography Conference}, pages 36--65. Springer,
  2015.

\bibitem[GPW18]{GPW18}
Mika G{\"{o}}{\"{o}}s, Toniann Pitassi, and Thomas Watson.
\newblock Deterministic communication vs. partition number.
\newblock {\em SIAM Journal on Computing}, 47(6):2435--2450, 2018.

\bibitem[GSS16]{GSS16}
Justin Gilmer, Michael Saks, and Srikanth Srinivasan.
\newblock Composition limits and separating examples for some boolean function
  complexity measures.
\newblock {\em Combinatorica}, 36(3):265--311, 2016.

\bibitem[Haj91]{H91}
P{\'e}ter Hajnal.
\newblock An Ω (n 4/3) lower bound on the randomized complexity of graph
  properties.
\newblock {\em Combinatorica}, 11(2):131--143, 1991.

\bibitem[Kin88]{K88}
Valerie King.
\newblock Lower bounds on the complexity of graph properties.
\newblock In {\em Proceedings of the twentieth annual ACM symposium on Theory
  of computing}, pages 468--476, 1988.

\bibitem[Kir74]{K74}
David Kirkpatrick.
\newblock Determining graph properties from matrix representations.
\newblock In {\em Proceedings of the sixth annual ACM symposium on Theory of
  computing}, pages 84--90, 1974.

\bibitem[KN97]{KN97}
Eyal Kushilevitz and Noam Nisan.
\newblock {\em Communication complexity}.
\newblock Cambridge University Press, 1997.

\bibitem[KT16]{KT16}
Raghav Kulkarni and Avishay Tal.
\newblock On fractional block sensitivity.
\newblock {\em Chicago J. Theor. Comput. Sci}, 8:1--16, 2016.

\bibitem[LZ17]{LZ17}
Chengyu Lin and Shengyu Zhang.
\newblock Sensitivity conjecture and log-rank conjecture for functions with
  small alternating numbers.
\newblock In Ioannis Chatzigiannakis, Piotr Indyk, Fabian Kuhn, and Anca
  Muscholl, editors, {\em 44th International Colloquium on Automata, Languages,
  and Programming, {ICALP} 2017, July 10-14, 2017, Warsaw, Poland}, volume~80
  of {\em LIPIcs}, pages 51:1--51:13. Schloss Dagstuhl - Leibniz-Zentrum
  f{\"{u}}r Informatik, 2017.

\bibitem[Mar58]{M58}
Andrey~A Markov.
\newblock On the inversion complexity of a system of functions.
\newblock {\em Journal of the ACM (JACM)}, 5(4):331--334, 1958.

\bibitem[Mid04]{M04}
Gatis Midrijanis.
\newblock Exact quantum query complexity for total boolean functions.
\newblock {\em arXiv preprint quant-ph/0403168}, 2004.

\bibitem[MO09]{M09}
Ashley Montanaro and Tobias Osborne.
\newblock On the communication complexity of xor functions.
\newblock {\em arXiv preprint arXiv:0909.3392}, 2009.

\bibitem[Mor09a]{M09a}
Hiroki Morizumi.
\newblock Limiting negations in formulas.
\newblock In {\em International Colloquium on Automata, Languages, and
  Programming}, pages 701--712. Springer, 2009.

\bibitem[Mor09b]{M09b}
Hiroki Morizumi.
\newblock Limiting negations in non-deterministic circuits.
\newblock {\em Theoretical Computer Science}, 410(38-40):3988--3994, 2009.

\bibitem[Nis91]{N91}
Noam Nisan.
\newblock Crew prams and decision trees.
\newblock {\em SIAM Journal on Computing}, 20(6):999--1007, 1991.

\bibitem[NS94]{NS94}
Noam Nisan and Mario Szegedy.
\newblock On the degree of boolean functions as real polynomials.
\newblock {\em Computational complexity}, 4(4):301--313, 1994.

\bibitem[NW95]{NW95}
Noam Nisan and Avi Wigderson.
\newblock On rank vs. communication complexity.
\newblock {\em Combinatorica}, 15(4):557--565, 1995.

\bibitem[RM97]{raz1997separation}
Ran Raz and Pierre McKenzie.
\newblock Separation of the monotone nc hierarchy.
\newblock In {\em Proceedings 38th Annual Symposium on Foundations of Computer
  Science}, pages 234--243. IEEE, 1997.

\bibitem[Ros73]{R73}
Arnold~L Rosenberg.
\newblock On the time required to recognize properties of graphs: A problem.
\newblock {\em ACM SIGACT News}, 5(4):15--16, 1973.

\bibitem[RV76]{RV76}
Ronald~L Rivest and Jean Vuillemin.
\newblock On recognizing graph properties from adjacency matrices.
\newblock {\em Theoretical Computer Science}, 3(3):371--384, 1976.

\bibitem[RY20]{RY20}
Anup Rao and Amir Yehudayoff.
\newblock {\em Communication Complexity: and Applications}.
\newblock Cambridge University Press, 2020.

\bibitem[ST04]{ST04}
Shao~Chin Sung and Keisuke Tanaka.
\newblock Limiting negations in bounded-depth circuits: An extension of
  markov's theorem.
\newblock {\em Information processing letters}, 90(1):15--20, 2004.

\bibitem[STV14]{STV14}
Amir Shpilka, Avishay Tal, and Ben~Lee Volk.
\newblock On the structure of boolean functions with small spectral norm.
\newblock In {\em Proceedings of the 5th conference on Innovations in
  theoretical computer science}, pages 37--48, 2014.

\bibitem[Sun07]{S07}
Xiaoming Sun.
\newblock Block sensitivity of weakly symmetric functions.
\newblock {\em Theoretical computer science}, 384(1):87--91, 2007.

\bibitem[SW86]{SW86}
Michael Saks and Avi Wigderson.
\newblock Probabilistic boolean decision trees and the complexity of evaluating
  game trees.
\newblock In {\em 27th Annual Symposium on Foundations of Computer Science
  (sfcs 1986)}, pages 29--38. IEEE, 1986.

\bibitem[SW93]{SC93}
Miklos Santha and Christopher Wilson.
\newblock Limiting negations in constant depth circuits.
\newblock {\em SIAM Journal on Computing}, 22(2):294--302, 1993.

\bibitem[Tal13]{tal2013properties}
Avishay Tal.
\newblock Properties and applications of boolean function composition.
\newblock In {\em Proceedings of the 4th conference on Innovations in
  Theoretical Computer Science}, pages 441--454, 2013.

\bibitem[TUR84]{T84}
G~TURIN.
\newblock The critical complexity of graph properties.
\newblock {\em Inform. Process. Lm}, 18:151--153, 1984.

\bibitem[TWXZ13]{THWX13}
Hing~Yin Tsang, Chung~Hoi Wong, Ning Xie, and Shengyu Zhang.
\newblock Fourier sparsity, spectral norm, and the log-rank conjecture.
\newblock In {\em 2013 IEEE 54th Annual Symposium on Foundations of Computer
  Science}, pages 658--667. IEEE, 2013.

\bibitem[vZGR97]{GRR97}
Joachim von Zur~Gathen and James~R Roche.
\newblock Polynomials with two values.
\newblock {\em Combinatorica}, 17(3):345--362, 1997.

\bibitem[Wat22]{W22}
Adam Wathieu.
\newblock Exposition of the {K}ushilevitz function.
\newblock {\em Technical report. Northwestern University, Computer Science
  Department}, NU-CS-2022-08, 2022.

\bibitem[Yao87]{Y87}
Andrew Chi-Chih Yao.
\newblock Lower bounds to randomized algorithms for graph properties.
\newblock In {\em 28th Annual Symposium on Foundations of Computer Science
  (sfcs 1987)}, pages 393--400. IEEE, 1987.

\end{thebibliography}

\newcommand{\etalchar}[1]{$^{#1}$}

\appendix
\section{Extended preliminaries}
\label[appendix]{sec:exprelims}
Throughout this section, we assume that $f:\{0,1\}^n\rightarrow \{0,1\}$is a generic Boolean function, and $x=(x_1,\ldots,x_n)\in\{0,1\}^n$ a generic input to $f$. Restriction of $f$ to a set $S \subseteq \{0,1\}^n$ is denoted by $f\mid_S$.
\begin{definition}[Subcube, co-dimension]
\label[definition]{def:subsube}
A set $C \subseteq \{0,1\}^n$ is called a \emph{subsube} is there exist a set of indices $\mathcal{I} \subseteq [n]$ and an assignment function $A: \mathcal{I} \rightarrow \{0,1\}$ such that $C=\{x \in \{0,1\}^n \mid \forall i \in I, x_i=A(i)\}$. Note that for every subcube $C$ there is a unique such $I$ and $A$. The \emph{co-dimension} of $C$, denoted by $\codim(C)$, is defined as $|I|$.
\end{definition}
Restriction of $f$ to a subcube $C$ is also simply referred to as \emph{a restriction of $f$}.
\begin{definition}[$\mathsf{TRIBES}$]
\label[definition]{def:tribes}
$\mathsf{TRIBES}_{\sqrt{n} \times \sqrt{n}}$ is a Boolean function on $n$ bits that is defined as follows. Let $x:=(x_{i,j})_{1\leq i \leq n, 1 \leq j \leq n} \in \left(\{0,1\}^{\sqrt{n}}\right)^{\sqrt{n}}$. Then,
\[\mathsf{TRIBES}_{\sqrt{n} \times \sqrt{n}}(x)=\vee_{i=1}^n (\wedge_{j=1}^n x_{i,j}).\]
\end{definition}
\begin{definition}[Indexing function]
\label[definition]{def:ind}
For $m\geq 1$, the indexing function $\ind_{m+2^m}$ is defined to be a function on $m+2^m$ bits defined as follows. Let $y\in\{0,1\}^m$ and $z = (z_0,\ldots, z_{2^m-1})\in \{0,1\}^{2^m}$. Let $\mathsf{bin}(y)$ denote the integer in $\{0\} \cup [2^m-1]$ whose binary expansion is $y$. Then
\[\ind_{m+2^m}(y,z)=z_{\mathsf{bin}(y)}.\]
\end{definition}
\begin{definition}[Majority function]
\label[definition]{def:maj} Let $n$ be an odd positive integer.
The majority function on $n$ bits is defined as follows.
\[\mathsf{MAJ}_n(x)=\left\{\begin{array}{cc}1 & \mbox{if $|\{i \in [n]\mid x_i=1\}|>\frac{n}{2}$,}\\0 & \mbox{otherwise.}\end{array}\right.\]
\end{definition}
\begin{definition}[Function composition]
\label{def:funcomp}
The composition of two Boolean functions $f: \{0,1\}^m \rightarrow\{0,1\}$ and $g: \{0,1\}^n \rightarrow\{0,1\}$, denoted as $f\circ g: \{0,1\}^{mn} \rightarrow \{0,1\}$ is defined as:

\[ (f\circ g) (x^{(1)}, x^{(2)}, \cdots x^{(m)}) = f(g(x^{(1)}), g(x^{(2)}), \cdots, g(x^{(m)})),\]

where for each $i\in [m]$, $x^{(i)} \in \{0,1\}^n$.
\end{definition}
\subsection{Decision tree complexity}
\begin{definition}[Deterministic decision tree complexity] 
\label[definition]{def:detcomp}
A deterministic decision tree $T$ is a rooted ordered binary tree. Each internal node of $T$ has two children and is labelled by an index in $[n]$. Each leaf is labelled $0$ or $1$. For an input $x=(x_1,\ldots,x_n)\in\{0,1\}^n$, $T$ is evaluated by starting from the root, and navigating down the tree till we reach a leaf, as follows. In a time step, let $i$ be the label of the current internal node. Then, we move to the left child of the current node if $x_i=0$ and its right child if $x_i=1$. When the computation reaches a leaf, the bit that labels the leaf is output. $T$ is said to compute $f$ if for each $x\in\{0,1\}^n$, $T$ outputs $f(x)$. The cost of $T$ on $x$, which we denote by $\mathsf{cost}(T, x)$, is the depth of the leaf of $T$ that $x$ takes it to.
The decision tree complexity of $f$ is defined as:

\[ \dclass(f) =  \min_{T} \max_x \mathsf{cost}(T,x),\]
where the minimum is over all decision trees that compute $f$, and the maximum is over all strings in $\{0,1\}^n$.
\end{definition}
Clearly $0 \leq \dclass(f) \leq n$.
\begin{definition}[Decision tree rank]
\label[definition]{def:rank}
Let $T$ be a decision tree. The \emph{rank} of each node $v$ of $T$, denoted by $\rank(v)$, is defined recursively as follows. If $v$ is a leaf, then $\rank(v)$ is $0$. Let $v$ be an internal node with children $v_\ell$ and $v_r$. If $\rank(v_\ell) \neq \rank(v_r)$ then $\rank(v)=\max\{\rank(v_\ell), \rank(v_r)\}$. Else, $\rank(v)=\rank(v_\ell)+1$. The rank of $T$, denoted by $\rank(T)$, is defined to be the rank of its root. the rank of $f$ is defined as
\[\rank(f)=\min_T \rank(T),\]
where the minimum is over all decision trees that compute $f$.
\end{definition}
It is easy to see that $\rank(f) \leq \dclass(f)$.
\begin{definition}[Parity decision tree complexity] 
\label[definition]{def:detparity}
A parity decision tree $T$ is a rooted ordered binary tree. Each internal node of $T$ has two children and is labelled by a subset of indices $S \subseteq [n]$. Each leaf is labelled $0$ or $1$. For an input $x=(x_1,\ldots,x_n)\in\{0,1\}^n$, $T$ is evaluated by starting from the root, and navigating down the tree till we reach a leaf, as follows. In a time step, let $S$ be the label of the current internal node. Then, we move to the left child of the current node if $\oplus_{i \in S}x_i=0$ and its right child if $\oplus_{i \in S}x_i=1$. When the computation reaches a leaf, the bit that labels the leaf is output. $T$ is said to compute $f$ if for each $x\in\{0,1\}^n$, $T$ outputs $f(x)$. The cost of $T$ on $x$, which we denote by $\mathsf{cost}(T, x)$, is the depth of the leaf of $T$ that $x$ takes it to.
The parity decision tree complexity of $f$ is defined as:

\[ \pdt(f) =  \min_{T} \max_x \mathsf{cost}(T,x),\]
where the minimum is over all parity decision trees that compute $f$, and the maximum is over all strings in $\{0,1\}^n$.
\end{definition}
\begin{definition}[Subcube decision tree complexity] 
\label[definition]{def:detsubcube}
A subcube decision tree $T$ is a rooted ordered binary tree. Each internal node of $T$ has two children and is labelled by a subcube $C \subseteq \{0,1\}^n$. Each leaf is labelled $0$ or $1$. For an input $x=(x_1,\ldots,x_n)\in\{0,1\}^n$, $T$ is evaluated by starting from the root, and navigating down the tree till we reach a leaf, as follows. In a time step, let $C$ be the label of the current internal node. Then, we move to the left child of the current node if $x \notin C$ and its right child if $x \in C$. When the computation reaches a leaf, the bit that labels the leaf is output. $T$ is said to compute $f$ if for each $x\in\{0,1\}^n$, $T$ outputs $f(x)$. The cost of $T$ on $x$, which we denote by $\mathsf{cost}(T, x)$, is the depth of the leaf of $T$ that $x$ takes it to.
The subcube decision tree complexity of $f$ is defined as:
\[ \subcubedt(f) =  \min_{T} \max_x \mathsf{cost}(T,x),\]
where the minimum is over all subcube decision trees that compute $f$, and the maximum is over all strings in $\{0,1\}^n$.
\end{definition}
\subsection{Complexity measures} 
\begin{definition}[Sensitivity]
\label[definition]{def:sensitivity}
An index $i \in [n]$ is said to be a \emph{sensitive} index (or a sensitive bit) for $x$ is $f(x)\neq f(x^{\oplus i})$. The \emph{sensitivity of $f$ on $x$}, denoted by $\s(f,x)$, is defined as the number of sensitive bits for $x$. In other words,
\[s(f,x)=|\{i \in [n] \mid f(x) \neq f(x^i)\}|.\]
The \emph{sensitivity of $f$}, denoted by $\s(f)$, is defined as
\[\s(f)=\max_{x \in \{0,1\}^n}\s(f,x).\]
\end{definition}
A set of indices $B \in [n]$ is called a \emph{sensitive block} of $x$ if $f(x)\neq f(x^{\oplus B})$.
\begin{definition}[Block sensitivity]
\label[definition]{def:bs}
The \emph{block sensitivity of $f$ on input $x$}, denoted by $\bs(f, x)$, is defined as the maximum number of disjoint sensitive blocks of $x$.

The \emph{block sensitivity of $f$}, denoted by $\bs(f)$, is defined as
\[\bs(f)=\max_{x\in\{0,1\}^n}\bs(f,x).\]
\end{definition}

Let $\{B_1,\ldots,B_t\}$ be the set of all sensitive blocks of $x$.  
$\bs(f,x)$ is the value of the following integer linear program.

\begin{equation*}
\begin{array}{lll@{}ll}
bs(f,x)=& \text{max}  & \displaystyle\sum_{j=1}^{t} w_{j} &\\
 &\text{subject to}& \displaystyle\sum_{j: B_j\ni i}   &w_{j} \leq 1  &\forall i\in [n]\\
       &          &                                                &w_{j} \in \{0,1\} & \forall j\in[t]
\end{array}
\end{equation*}
The value of the linear program obtained by relaxing the integrality constraint of the above  program is called the fractional block sensitivity of $f$ on input $x$, denoted by $\fbs(f, x)$.
\begin{definition}[Fractional block sensitivity \cite{KT16, GSS16}]
\label[definition]{def:fbs}
The \emph{fractional block sensitivity} of $f$ on $x$, denoted by $\fbs(f,x)$, is defined as
\begin{equation*}
\begin{array}{lll@{}ll}
fbs(f,x)=& \text{max}  & \displaystyle\sum_{j=1}^{t} w_{j} &\\
 &\text{subject to}& \displaystyle\sum_{j:B_j\ni i}   &w_{j} \leq 1  &\forall i\in [n]\\
       &          &                                                &w_{j} \in [0,1] & \forall j\in[t]
\end{array}
\end{equation*}
The \emph{fractional block sensitivity of $f$}, denoted by $\fbs(f)$, is defined as 
\[\fbs(f)=\max_{x\in\{0,1\}^n}\fbs(f,x).\]
\end{definition}
\begin{definition}[Certificate complexity, minimum certificate complexity and maximin certificate complexity]
\label{def:certificate}
A subcube $C$ is called a \emph{$0$-certificate} (resp. $1$-certificate) of $f$ if $f\mid_C$ is the constant $0$ (resp. $1$) function. A subcube is called a \emph{certificate} if it is a $0$-certificate or a $1$-certificate. The \emph{certificate complexity of $f$ on $x$}, denoted by $\cer(f,x)$, is the smallest co-dimension of a certificate $C$ that contains $x$. For a bit $b\in\{0,1\}$ the \emph{$b$-certificate complexity of $f$}, denoted by $\cer_b(f)$, is defined as
\[ \cer_b(f) = \max_{x \in \{0,1\}^n, f(x)=b } \cer(f,x).\]
The \emph{certificate complexity of $f$}, denoted by $\cer(f)$, is defined as,
\[ \cer(f) = \max \{\cer_0(f), \cer_1(f)\}.\]
A \emph{minimum certificate} of $f$ is defined to be a subcube of minimum co-dimension on which $f$ is constant.
The \emph{minimum certificate complexity of $f$}, denoted by $\cmin(f)$, is defined as the co-dimension of a minimum certificate. Equivalently,
\[ \cmin(f) = \min_{x \in \{0,1\}^n } \cer(f,x).\]
The \emph{maximin certificate complexity of $f$}, denoted by $\cminstar(f)$, is defined as,
\[ \cminstar(f) = \max_C \cmin(f\mid_C).\]
Above, the maximum is over all subcubes of $\{0,1\}^n$, and $f\mid_C$ is viewed as a Boolean function on the Boolean hypercube $\{0,1\}^{[n]\setminus I}$.
\end{definition}
We say that a set of indices $S\subseteq [n]$ forms a certificate for $x$, to mean that $\{y \in \{0,1\}^n \mid y_i=x_i \forall i \in S\}$ is a certificate.
\begin{definition}[$P_f$ and degree]
It is well known that every Boolean function $f:\{0,1\}^n\rightarrow\{0,1\}$ can be represented by a unique multi-linear polynomial $P_f(x_1,\ldots,x_n)$ with real coefficients such that $P(x)=f(x)$ for all $x \in \{0,1\}^n$. The \emph{exact degree of $f$}, or simply the degree of $f$, denoted by $\degree(f)$, is defined to be the degree of $P_f$. 
\end{definition}
\paragraph{Identification of variables} Let $i,j \in [n]$ be distinct indices in $[n]$. Let $S:=\{x \in \{0,1\}^n \mid x_i=x_j\}$. The restriction $f\mid_S$ of $f$ to $S$ is the function obtained from $f$ by identifying variables $x_i$ and $x_j$. $f\mid_S$ can be thought of as a Boolean function defined on the hypercube $\{0,1\}^{[n]\setminus\{j\}}$.
\begin{definition}
\label[definition]{def:symfun}
$f$ is a \emph{symmetric} Boolean function if $f(x)$ depends only on the Hamming weight of $x$, i.e., there exists a function $g:[n]\cup\{0\} \rightarrow \{0,1\}$ suh that $f(x)=g(|x|)$.
\end{definition}
The following proposition lists some well-known facts about various complexity measures \cite{BdeW02Survey, KT16, GSS16}.
\begin{proposition}
\label[proposition]{prop:basic_facts}\  
\begin{enumerate}
\item $s(f,x)\leq \bs(f,x) \leq \fbs(f,x) \leq \cer(f,x)\leq \dclass(f)$.
\item $\cer(f)=O(\bs(f)^2)$. \cite{N91}
\item $\bs(f)=O(\degree(f)^2)$. \cite{NS94}
\item $\dclass(f)=O(\cer(f)\cdot\bs(f))$. \cite{BBCMdW01}
\item $\dclass(f)=O(\degree(f)\cdot\bs(f))$. \cite{M04}
\item $\degree(f)\leq\dclass(f)$.
\item Measures $\s, \bs, \fbs, \cer, \cminstar, \degree, \dclass$ do not increase under restrictions to subcubes. That is, for any $\mathsf{M} \in \{\s, \bs, \fbs, \cer, \cminstar, \degree, \dclass\}$ and any subcube $C$, $\mathsf{M}(f\mid_C)\leq \mathsf{M}(f).$\footnote{Note that the same cannot be said about$\cmin$.}
\item $\bs, \degree, \cer$ and $\dclass$ do not increase under identification of variables. That is, for $\mathsf{M} \in \{\bs, \degree\, \cer, \dclass\}$, distinct $i,j \in [n]$ and $S=\{x \in \{0,1\}^n \mid x_i=x_j\}$,  $\mathsf{M}(f\mid_S)\leq \mathsf{M}(f)$.\footnote{The same cannot be said about $\s$.}
\item Let $f$ be non-constant and symmetric. Then $\degree(f)=n-o(n)$ \cite{GRR97} and $\bs(f)=\s(f)\geq\lceil\frac{n+1}{2}\rceil$ \cite{T84}.
\end{enumerate}
\end{proposition}
\subsection{Communication complexity}
In the two-party communication model (introduced by Yao \cite{A79}), two parties Alice and Bob jointly hold an input to a Boolean function $f:\{0,1\}^{n_1} \times \{0,1\}^{n_2} \rightarrow \{0,1\}$. There is a bi-partition of  the set of input indices $[n]$ into two parts of sizes $n_1$ and $n_2$, say, where $n_1+n_2=n$. Alice and Bob respectively hold the input bits that correspond to the indices in these two parts. Let $x^{(1)}$ and $x^{(2)}$ be the input strings held by Alice and Bob respectively. Alice and Bob are interested in computing $f(x^{(1)},x^{(2)})$, and they are allowed to exchange bits in an interactive fashion via a communication channel. The objective is to jointly compute $f(x^{(1)},x^{(2)})$ by exchanging as few bits as possible. A set of rules that the parties follow to determine the messages that they send to each other in each step is called a protocol. A protocol for $f$ is correct if for every input $(x^{(1)},x^{(2)})$, at the end of the interaction, a party outputs $f(x^{(1)},x^{(2)})$.

The cost of a protocol $\Pi$ on an input $(x^{(1)},x^{(2)})\in \{0,1\}^{n_1} \times \{0,1\}^{n_2}$, denoted by $\mathsf{cost}(\Pi, (x^{(1)},x^{(2)}))$, is the number of bits exchanged by $\Pi$ on input $(x^{(1)},x^{(2)})$. The deterministic communication complexity of $f$, denoted by $\CC(f)$, is defined as follows:

\[\CC(F) = \min_{\Pi} \max_{(x^{(1)},x^{(2)})\in \{0,1\}^{n_1} \times \{0,1\}^{n_2}} \mathsf{cost}(\Pi, (x^{(1)},x^{(2)})).\]
where the minimum is over all correct protocols of $f$. The bi-partition of the inputs in the definition of $\CC$ is implicit and will be clear from the context. See Section~\ref{prelims} and textbooks by Rao and Yehudayoff \cite{RY20} and by Kushilevitz and Nisan \cite{KN97} for a comprehensive introduction to the subject.
\section{Useful facts about zebra functions}
\label[appendix]{sec:zebra-facts}
In this section, we prove \Cref{prop:zebra-useful-facts}.
\begin{proof}[Proof of \Cref{prop:zebra-useful-facts}]
 
(Part 1) Follows immediately from the definition of stripes.

(Part 2) Suppose without loss of generality that $f(0^n)=0$. Now, one can show that $x\in\{0,1\}^n$, $f(x)=\alt_f(x) \mod 2$ by an easy induction on $|x|$ ($|x|+1$ is the number of points on any monotone path from $0^n$ to $x$). This completes the proof of this part.

(Part 3) Consider set $S=\{z: z \le \mbox{ and }x\in \stripe_f(i)\}$. $S \neq \phi$ as $x \in S$. The proof follows by considering a point in $S$ with minimum hamming weight.

(Part 4) Without loss of generality, let $y$ be obtained by flipping a $1$-bit of $x$ to $0$.
Since $x$ is a minimal point of $\stripe_f(i)$ we have $y \notin \stripe_f(i)$. Considering any monotone path from $0^n$ to $x$ through $y$, we conclude that $\alt_f(x)=\alt_f(y)+1$. the proof follows.

(Part 5) For any string $x\in \{0,1\}^{[n]\setminus\{j\}}$, define a string $e_j(x)$ to be the following string in $\{0,1\}^n$:
\[\forall k \in [n], \left(e_j(x)\right)_k=\left\{\begin{array}{cc}x_k & \mbox{if $k \neq j$,}\\b & \mbox{otherwise.}\end{array}\right.\]
Then, for all $y \in \{0,1\}^{[n]\setminus \{j\}}$, $f\mid_{x_j=b}(y)=f(e_j(y))$. Thus the alternation number of a monotone path from $0^{[n]\setminus\{j\}}$ to $1^{[n]\setminus\{j\}}$ with respect to $f\mid_{x_j=b}$ is equal to the alternation number of some monotone path from $e_j(0^{[n]\setminus\{j\}})$ to $e_j(1^{[n]\setminus\{j\}})$ with respect to $f$. The proof now follows by \Cref{prop:groundwork-for-stripes}.
\end{proof}

\section{Query complexity and maximin certificate complexity}
\label[appendin]{sec:cminstar_key}
In this section we prove \Cref{lemma:d_n_cstar}.
\begin{proof}[Proof of \Cref{lemma:d_n_cstar}]
Consider the following decision tree computing $f$.\\
\begin{algorithm}[H]
\label{algo: decision-tree}
\SetAlgoLined
\textbf{Input: }$x=(x_1,\ldots,x_n)\in\{0,1\}^n$.\\
$\mathcal{D} \gets \{0,1\}^n, g \gets f$.\\
\While{$g$ is not constant}{
    Let $C=\{y \in \mathcal{D} \mid y_i=a_i \forall i \in S\}$ be a certificate of $g$ with least co-dimension.\\
    Query $x_i$ for all $i\in S$. let the outcomes be $x_i=a'_i$ for all $i \in S$.\\
    $\mathcal{D} \gets \mathcal{D} \cap \{y \in \mathcal{D} \mid y_i=a'_i \forall i\in S\}$.\\
    $g\gets g\mid_\mathcal{D}$.
}
Return value of constant function $g$.
\caption{}
\end{algorithm}

Algorithm~\ref{algo: decision-tree} is easily seen to always halt and compute $f$.

(Part 1) We will show that Algorithm~\ref{algo: decision-tree} makes at most $\cminstar(f)(\bs^{(0)}(f)+\bs^{(1)}(f)) \le 2\cminstar(f)\bs(f)$ queries.

Towards a contradiction, suppose if possible that the number of iterations is strictly more than $\bs^{(0)}(f)+\bs^{(1)}(f)$ for some input.
Then Algorithm~\ref{algo: decision-tree} queries either at least $\bs^{(1)}(f)+1$ 0-certificates or at least $\bs^{(0)}(f)+1$ 1-certificates. 
We assume without loss of generality that the algorithm queries at least $\bs^{(1)}(f)+1$ 0-certificates. Consider the domain $\mathcal{D}$ right after the $\bs^{(1)}(f)$-th 0-certificate is queried. Since $f\mid_\mathcal{D}$ is not a constant function there exits an input $z \in \mathcal{D}$ such that $f(z) = 1$. Now, for each 0-certificate $C$ queried so far, the set of indices where $z$ and $C$ disagree is a sensitive block of $z$. Furthermore, these blocks are all disjoint, as each time the algorithm chooses a certificate that is consistent with the query outcomes so far. It follows that the algorithm has already fully queried $\bs(f)\geq \bs(f,z)$ may disjoint blocks of $z$, and the answers are all consistent with $z$. Since these bits form a 1-certificate, we have reached a contradiction to the fact the $f\mid_{\mathcal{D}}$ is non-constant. This completes the proof.

(Part 2) We will show that Algorithm~\ref{algo: decision-tree} makes at most $\cminstar(f)\degree(f)$ queries.

First we show that every certificate of $f$ must intersect every maximum degree monomial of the unique multilinear polynomial $P_f$ that represents $f$. Suppose not, and there exists a certificate $C=\{y \in \mathcal{D} \mid y_i=a_i \forall i \in S\}$ and a maximum monomial $M$ of $P_f$, such that $S \cap M=\emptyset$. For each $i\in S$, substitute $a_i$ for $y_i$. since $M$ is a leading monomial and disjoint from $S$, the resultant polynomial is non-constant, that contradicts the assumption that $C$ is a certificate.

It follows that after every iteration of the \emph{while} loop, $\degree(g)$ drops by at least 1.
Hence the number of iterations of the loop is at most $\degree(f)$.
Since we make at most $\cminstar(f)$ queries in each iteration, the theorem follows.
\end{proof}
\section{Decision tree rank and subcube decision tree}\label{section-proof-of-lifting}
In this section we prove \Cref{proposition:lifting}.
\begin{proof}{of \ref{proposition:lifting} (\ref{proposition:lifting1})}
(Part (1)) The first inequality follows from the observation that a communication protocol can simulate a subcube decision tree of depth $d$ by exchanging at most $2d$ bits. Each query of a subcube decision tree can be evaluated by an exchange of 2 bits, as follows. Let the subcube queried at an internal node be $C=\{x \in \{0,1\}^n \mid x_i=b_i \forall i \in S\}$. For each $i\in S$, $x_i$ is held by one of the parties. Thus, each party may separately check if the values of all the variables $x_i$ for $i\in S$ that they hold match with $a_i$, exchange the outcomes of those checks, and thus answer the subcube query of membership in $C$. The second inequality holds as a decision tree is a special kind of subcube decision tree where all the subcubes queried have co-dimension $1$.

(Part 2) For the first inequality, we prove by induction on $t$ that for every subcube decision tree $T$ of depth $t$, there exists a decision tree $T'$ of rank at most $t$ that computes the same function. The base case $t=0$ is trivial. Now, let $X$ be a subcube decision tree of $f$ of depth $t\geq 1$. Let $a$ be its root, querying subcube $\{y \in \{0,1\}^n \mid y_{i_1}=a_{i_1},\ldots,y_{i_k}=a_{i_k}\}$, and the two subtrees of $a$ be $X_1$ and $X_2$, each of depth at most $t-1$. By inductive hypothesis, there exist two decision trees $T_1$ and $T_2$, each of rank at most $t-1$, that compute same functions as $X_1$ and $X_2$ respectively. Next, we replace $a$ by a decision tree $T''$ that queries the input variables with indices $i_1,\ldots,i_k$ and decides on the outcome of the subcube query. $T''$ queries the variables in order, and outputs $0$ if it finds a $i_j$ such that $y_{i_j}\neq a_{i_j}$. If it finds no such $i_j$ it outputs $1$. Next, we replace each leaf of $T''$ by either $X_1$ or $X_2$ depending on the outcome of the subcube query in that leaf. It is easy to check that the rank of the resulting tree is at most $(t-1)+1=t$.

Now we prove the second inequality. Let $T$ be a decision tree with rank $\rank(f)=r\geq 1$, say, and depth $d \leq n$, that computes $f$. We will show that there is a subcube decision tree $T'$ of depth at most $r (\log n+1)$ that computes $f$.

Let $v$ be any node of $T$. Let the variables queried on the unique path from the root of $T$ to $v$ be $x_{i_1},\ldots,x_{i_\ell}$ and the answers to these queries be $a_{i_1},\ldots,a_{i_\ell}$. Thus, if $T$ is run on an input $x$, the computation visits node $v$ \emph{if and only if} $x$ belongs to the subcube $\{y \in \{0,1\}^n \mid y_{i_1}=a_{i_1},\ldots,y_{i_k}=a_{i_k}\}$. Thus, a single subcube query suffices to determine whether a particular node of $T$ is visited. Extending this idea, we may do a binary search on the unique path in $T$ from its root to $v$ to find out the first vertex off the path that the computation visits when $T$ is run on $x$. The number of subcube queries spent is at most one more the logarithm of the depth of $v$; in particular, it is at most $\log n+1$.

Now, we describe $T'$ recursively. Let $x$ be an input. Let $v$ be a deepest vertex in $T$ such that $\rank(v)=r$. Let $P$ denote the path in $T$ from its root to $v$. Using the idea of the preceding paragraph $T'$ first finds out, spending at most $\log n+1$ queries, the first vertex $u \notin P$ in $T$ that the computation visits when $T$ is run on $x$. Next, we show that $\rank(u)\leq r-1$. To see this, first note that $u$ is a child of a vertex $w\in P$. Now, if $w \neq v$, then $w$ has a child that is in $P$ and hence has rank $r$. Since $\rank(w)=r$, therefore the rank of the other child $u$ of $w$ must be strictly less than $r$. On the other hand, if $w=v$, then $\rank(u)\leq r-1$ by the definition of $v$. Next, $T'$ recursively runs the subtree of $T$ rooted at $u$ on $x$, whose rank is at most $r-1$ by the preceding argument. The complexity of $T'$ is readily seen to be at most $r\cdot(\log n+1)$.
\end{proof}
\end{document}